\documentclass[11pt,reqno]{amsart}
\usepackage{amsfonts,amssymb,amsmath,amsopn,amsthm,graphicx}
\usepackage{amsxtra, mathrsfs}
\usepackage{colordvi}
\usepackage[usenames,dvipsnames]{color}
\usepackage{amsfonts,amssymb,amsbsy,amsmath,amsthm,dsfont}
\usepackage{mathtools}
\usepackage[colorlinks=true]{hyperref}
 
\usepackage{enumitem}

\usepackage[parfill]{parskip}    

\usepackage[parfill]{parskip}    

\numberwithin{equation}{section}

\newtheorem{theorem}{Theorem}[section]
\newtheorem{corollary}[theorem]{Corollary}
\newtheorem{lemma}[theorem]{Lemma}
\newtheorem{proposition}[theorem]{Proposition}

\theoremstyle{definition}

\newtheorem{remark}[theorem]{Remark}

\newcommand{\B}{\mathscr{B}}

\newcommand{\C}{\mathbb{C}}
\newcommand{\CC}{\mathcal{C}}

\renewcommand{\epsilon}{\varepsilon}

\newcommand{\h}{\mathcal{H}}
\newcommand{\Q}{\mathcal{Q}}
\newcommand{\V}{\mathcal{V}}

\newcommand{\N}{\mathbb{N}}

\renewcommand{\phi}{\varphi}
\newcommand{\PP}{\mathbb{P}}
\newcommand{\R}{\mathbb{R}}
\newcommand{\U}{\mathcal{U}}

\newcommand{\Sph}{\mathbb{S}}

\newcommand{\Z}{\mathbb{Z}}
\newcommand{\eps}{\epsilon}

\newcommand{\id}{\mathds{1}}

\newcommand{\m}{{\scriptscriptstyle-}}
\newcommand{\pp}{{\scriptscriptstyle+}}
\newcommand{\ppm}{{\scriptscriptstyle\pm}}

\DeclareMathOperator{\tr}{Tr}

\begin{document}

\title[CLR-type estimates for  the Pauli and magnetic Schr\"odinger operator ]{Cwikel-Lieb-Rozenblum type estimates for the Pauli and magnetic Schr\"odinger operator in dimension two}

\author {Matthias Baur}
\address [Matthias Baur]{Institute of Analysis, Dynamics and Modeling, Department of Mathematics, University of
Stuttgart, Pfaffenwaldring 57, 70569 Stuttgart, Germany}
\email {matthias.baur@mathematik.uni-stuttgart.de}

\author {Hynek Kova\v{r}\'{\i}k}
\address [Hynek Kova\v{r}\'{\i}k]{DICATAM, Sezione di Matematica, Universit\`a degli studi di Brescia, Italy}
\email {hynek.kovarik@unibs.it}

\begin{abstract}
We prove a Cwikel-Lieb-Rozenblum type inequality for the number of negative eigenvalues of Pauli operators in dimension two. The resulting upper 
bound is sharp both in the weak as well as in the strong coupling limit. We also derive different upper bounds for magnetic Schr\"odinger operators. The nature 
of the two estimates depends on whether or not the spin-orbit coupling is taken into account.

\end{abstract}

\maketitle


\section{\bf Introduction and main results}
\label{sec-intro}
\subsection{Motivation} 
\label{ssec-bg} 
The famous Cwikel-Lieb-Rozenblum (CLR) inequality says that the number $N(-\Delta -V)$ of negative eigenvalues of a Schr\"odinger operator $-\Delta-V$ in $L^2(\R^d)$, the so-called counting function,
satisfies, for $d\geq 3$, the upper bound 
\begin{equation} \label{clr-rn}
N(-\Delta -V) \, \leq \, C_d\int_{\R^d} V_+(x)^{\frac d2}\, dx,
\end{equation}
where $C_d$ is a constant which depends only on the dimension and where $V_\pm :=\max (0,\pm V)$. The  inequality was proved independently by Cwikel, Lieb and Rozenblum in
\cite{cw}, \cite{lieb} and \cite{roz}. See also the recent paper \cite{hkrv}. For further background and reading we refer to the monograph \cite{flw-book} and references therein.

Here we treat the case $d=2$. The well-known phenomenon of weakly coupled eigenvalues \cite{sim}, absent in dimensions $d\geq 3$, implies that \eqref{clr-rn} must fail. 
As a replacement, upper bounds of the form 
\begin{equation} \label{clr-2-dim}
N(-\Delta -V) \, \leq \, 1 + G[V] , \qquad d=2,
\end{equation}
with certain homogenous functionals $G[\, \cdot\, ]$ of degree one were obtained in \cite{sol, chad, new,  st, kvw, shar}. Note that the existence of potentials $V\in L^1(\R^2)$ which 
induce a super-linear growth of $N(-\Delta -\lambda V)$ in $\lambda$ as $\lambda\to\infty$, see \cite{bl}, forbids us to put $G[V] = C \int_{\R^2} V_\pp(x)\, dx$, which would be a natural extension of \eqref{clr-rn}. 
Instead, the functional $G$ often includes weighted integrals of  $V_\pp$. To make an example let us mention that for radial potentials 
\begin{equation} \label{G-radial} 
G[V] = C \int_{\R^2}\, V_\pp(|x|)\, (1+|\log |x||)\, dx  \qquad \text{if}\quad  V(x)=V(|x|).
\end{equation} 
On the other hand, 
the fact that $G[\, \cdot\, ]$ is homogeneous of degree one implies 
that the bound \eqref{clr-2-dim} is, for a wide class of potentials, order-sharp for $\lambda\to\infty$. In fact, the Weyl asymptotic formula states that if $V$ is  continuous and compactly supported, then
\begin{equation} \label{weyl-2-dim}
\lim_{\lambda\to \infty}  \lambda^{-1}\, N(-\Delta -\lambda V)  = \frac{1}{2\pi}\, \int_{\R^2} V_+(x)\, dx,
\end{equation}
see e.g.~\cite[Sec.~4.4]{flw-book}. Moreover, \eqref{clr-2-dim} is sharp also in the weak coupling regime $\lambda\to 0$. Indeed, if $\int_{\R^2} V>0,$ then the operator $-\Delta -\lambda V$ has 
for $\lambda>0$ and small enough exactly one negative eigenvalue, \cite{sim, hhrv}. Put differently,
\begin{equation} \label{weak-schr-no-mag}
\lim_{\lambda\to 0+}   N(-\Delta -\lambda V)  = 1\, .
\end{equation}
The upper bound \eqref{clr-2-dim}  thus provides a valid alternative of the CLR-inequality in dimension two.

\smallskip

In this paper we prove an analog of the CLR-inequality for non-relativistic  two-dimensional magnetic Hamiltonians. Namely, for the Pauli and for the magnetic 
Schr\"odinger operator. The former acts in $L^2(\R^2,\C^2)$ and is formally given by
\begin{equation}\label{pauli-operator}
\PP =  \begin{pmatrix}
		H_+ & 0 \\
		0  &   H_-
	\end{pmatrix}\,  , \qquad H_\ppm =  (i\nabla +A)^2  \pm  B\,.
\end{equation}
Here $A:\R^2\to\R^2$ is a vector field and the function $B:\R^2\to\R$ is defined by
$$
B = {\rm curl}\, A = \partial_1 A_2 - \partial_2 A_1 \,.
$$
Since 
$$
\PP =\big(\sigma \cdot (i\nabla +A) \big)^2, \quad \sigma =( \sigma_1, \sigma_2),\quad  \sigma_1=
 \begin{pmatrix}
		0 & 1 \\
		1  &  0
	\end{pmatrix}\,  
\quad 
\sigma_2=  \begin{pmatrix}
		0 & i \\
		i  &  0
	\end{pmatrix}\,  ,
$$
it follows that $\PP\geq 0$. The matrix structure of $\PP$  comes from the spin-orbit coupling with the magnetic field, reflected by the $\pm$ sign of $B$ in $H_\ppm$. We restrict ourselves to the case where
 $\PP$ is 
perturbed by a scalar potential $V:\R^2\to\R$, to be understood as a multiplication by $V$ times the identity matrix $\id$. 
Hence our goal is to find an upper bound on $N(\PP -V)$ in  $L^2(\R^2,\C^2)$.

The regularity and decay conditions on $B$ are stated in \eqref{ass-B} below. In particular, the latter ensures that the magnetic field 
produces finite (normalized) flux 
\begin{equation} \label{flux}
\alpha:= \frac{1}{2\pi} \int _{\R^2}B(x)\, dx < \infty \, .
\end{equation}
Assume now for simplicity that $V$ is bounded, compactly supported and non-negative.
In this situation the number of negative eigenvalues of $\PP-\lambda V$, for $\lambda >0$ small enough, is positive and depends on  $\alpha$. More 
precisely, by \cite[Thm.~10.1]{weidl} we have
\begin{equation} \label{weak-pauli}
\lim_{\lambda\to 0+}   N(\PP  -\lambda V)  =  m(\alpha), 
\end{equation}
where 
\begin{equation}  \label{m-alpha}
m(\alpha)=  \max\{1+ [\, |\alpha|\, ] , \, 2\} 
\end{equation}
with  $[\, |\alpha|\, ] =\max\{k\in\N: k\leq |\alpha| \}$ being the integer part of $|\alpha|$. By the Aharonov-Casher theorem, \cite{ac, cfks, weidl}, if $|\alpha| >1$, then zero is an eigenvalue
of $\PP$ and its multiplicity is equal to $ [\, |\alpha|\, ]$ if $\alpha\not\in\Z$, and to  $ |\alpha| -1$ otherwise. Equation \eqref{weak-pauli} thus reflects the fact that in addition to the 
zero energy eigenfunctions, the Pauli operator admits also two, respectively one, virtual bound states (depending on whether or not $\alpha\in\Z$), i.e.~bounded solutions 
to the equation $\PP\, u=0$ such that $u\not\in L^2(\R^2; \C^2)$, \cite{weidl}.

The picture  changes completely if the spin-orbit coupling is neglected. In this case the matrix structure of the Pauli operator is destroyed and the Hamiltonian  \eqref{pauli-operator} reduces to two copies of the scalar magnetic Laplacian $(i\nabla +A)^2$ in $L^2(\R^2)$.  Accordingly, the absence of the spin-orbit coupling leads to the stabilization of the spectrum under small perturbations, see \cite{lw,weidl}. In other  words, in the weak coupling limit we have
\begin{equation} \label{weak-schr}
\lim_{\lambda\to 0+}   N\big ( (i\nabla +A)^2-\lambda V\big )  = 0.
\end{equation}
Contrarily, in the strong coupling regime, when $\lambda\to \infty$, there is typically no difference between the leading order terms of $N\big ( (i\nabla +A)^2-\lambda V\big )$ and $N(H_\pm  -\lambda V)$; all counting functions obey, for a generic $V$, the Weyl law, cf.~\eqref{strong-coupling}. Since $N(\PP - V)= N(H_+  -\lambda V) + N(H_-  -\lambda V)$, the counting function of the Pauli operator behaves identically except for a factor of two.

Any adequate  upper bound on $N(\PP - V)$ and $N\big ( (i\nabla +A)^2- V\big )$ thus should reflect the asymptotic behavior of these quantities both in weak and strong coupling regime displayed by equations \eqref{weak-pauli}, \eqref{weak-schr} and \eqref{strong-coupling}.

\subsection{Assumptions and main results}  To formulate our main results we need to introduce some necessary notation. Throughout the paper we will work under the following generic
assumption on the magnetic field:
\begin{equation}  \label{ass-B}
B\in L^q_{\rm\, loc} (\R^2), \ q>2, \qquad\ \  |B(x)| = \mathcal{O}(|x|^{-2-\eps}), \ \  \eps>0, \quad |x|\to\infty.
\end{equation}  
The assumptions on $V$ differ from case to case and will be specified later.

We introduce the function class $L^1(\R_+, L^p(\Sph))$ defined in polar coordinates on $\R^2$ as follows;
\begin{equation} \label{L-1p}
L^1(\R_+, L^p(\Sph)) = \Big\{ f: \R^2\to \C: \, \int_0^\infty \left(\int_0^{2\pi} |f(r,\theta)|^p\, d\theta\right)^{\frac 1p}\, r\, dr <\infty \Big\}\, .
\end{equation}
For  the associated norm we will adopt the shorthand
\begin{equation} \label{1p-norm}
\|f\|_{1,p} := \int_0^\infty \left(\int_0^{2\pi} |f(r,\theta)|^p\, d\theta\right)^{\frac 1p}\, r\, dr .
\end{equation}
By $\B_R=\{x\in\R^2: \, |x|< R\}$ we denote the ball of radius $R$ centered at the origin.  The indicator function of a set $M$ is denoted by $\id_M$.

The following theorem is the main result of our paper. 

\begin{theorem}[\bf Pauli operators]\label{thm-main-pauli}
Let $B$ satisfy Assumption  \eqref{ass-B}, and recall that $\alpha$ is given by \eqref{flux}. 
\begin{enumerate} 
\item {\bf Local logarithmic correction}. Assume that $\alpha \not\in\Z$. Then for any $p>1$ there exist constants $C_1=C_1(B,p)$ and $C_2=C_2(B)$ such that 
\begin{equation}  \label{clr-pauli-1}
N(  \PP -   V) \ \leq\     m(\alpha)+ C_1\, \|V_\pp\|_{1,p} +C_2 \| V_\pp  \log |x|\|_{L^1(\!\B_1)} 
\end{equation}
for all  $V \in L^1(\R_+, L^p(\Sph))$ with  $V  \log |\cdot | \in L^1(\B_1)$. Recall that $m(\alpha)$ is defined in \eqref{m-alpha}.

\bigskip

\item {\bf Global logarithmic correction}. Assume that $\alpha \in\Z$. Then for any $p>1$ there exist constants $\CC_1=\CC_1(B,p)$ and $\CC_2=\CC_2(B)$ such that 
\begin{equation}   \label{clr-pauli-2}
N(  \PP  -   V) \ \leq\  m(\alpha) + \CC_1\, \|V_\pp\|_{1,p} +\CC_2 \| V_\pp\log |x|\|_{L^1(\R^2)} 
\end{equation}
for all  $V \in L^1(\R_+, L^p(\Sph))$ with  $V  \log |\cdot | \in L^1(\R^2)$.
\end{enumerate}
\end{theorem}
Theorem \ref{thm-main-pauli} follows from Propositions \ref{prop-H-} and \ref{prop-H+} which are proved in Sections \ref{sec-upperb-minus} respectively \ref{sec-upperb-plus}. 
For radial potentials we get

\begin{corollary}[\bf Radial potentials]\label{cor-radial-pauli}
Let $B$ satisfy Assumption  \eqref{ass-B} and suppose that $V(x)=V(|x|)$.
\begin{enumerate} 
\item Assume that $\alpha \not\in\Z$. Then there exists a constants $C=C(B)$ such that 
\begin{equation}  \label{clr-pauli-radial-1}
N(  \PP -   V) \ \leq\     m(\alpha)+ C  \int_{\R^2} V_\pp(|x|) \big(1+ \id_{\{|x|<1\}} | \log |x|| \big)\, dx
\end{equation}
for all  $V \in L^1(\R^2)$ with  $V  \log |\cdot | \in L^1(\B_1)$.

\bigskip

\item  Assume that $\alpha \in\Z$. Then there exists a constant $\CC=\CC(B)$  such that 
\begin{equation}   \label{clr-pauli-radial-2}
N(  \PP -   V) \ \leq\     m(\alpha)+ \CC \int_{\R^2} V_\pp(|x|) \big(1+  | \log |x|| \big)\, dx
\end{equation}
for all  $V \in L^1(\R^2)$ with  $V  \log |\cdot | \in L^1(\R^2)$.
\end{enumerate}
\end{corollary}

\smallskip

Another consequence of Theorem \ref{thm-main-pauli}, or rather of its proof,  is the following bound on the number of  negative eigenvalue of magnetic Schr\"odinger operators.

\begin{corollary}[\bf magnetic Schr\"odinger operators]\label{cor-main-schr}
Let $B$ satisfy \eqref{ass-B} and assume that $\alpha\neq 0 $. Let $p>1$ and let $C_j$ and  $\CC_j$ be  the constants in Proposition \ref{prop-H+}.  

\begin{enumerate}
\item {\bf Local logarithmic correction}. If $\alpha\not\in\Z$, then 
\begin{equation}   \label{clr-schr-1}
N( (i\nabla +A)^2 -   V) \ \leq\  2\, C_1\, \|V_\pp\|_{1,p} +2\,  C_2\,  \| V_\pp  \log |x|\|_{L^1(\!\B_1)} 
\end{equation}
for all  $V \in L^1(\R_+, L^p(\Sph))$ with  $V  \log |\cdot | \in L^1(\B_1)$.

\bigskip

\item {\bf Global logarithmic correction}. If $\alpha \in\Z$, then
\begin{equation}  \label{clr-schr-2}
N( (i\nabla +A)^2 -   V) \ \leq\   2\,  \CC_1\, \|V_\pp\|_{1,p}  +2\, \CC_2\,  \| V_\pp\log |x|\|_{L^1(\R^2)} 
\end{equation}
for all  $V \in L^1(\R_+, L^p(\Sph))$ with  $V  \log |\cdot | \in L^1(\R^2)$.
\end{enumerate}
\end{corollary}
The proof of Corollary \ref{cor-main-schr} is given in Section \ref{sec-schr}.

Our proofs imply explicit bounds on all the constants involved in Theorem \ref{thm-main-pauli} and Corollary \ref{cor-main-schr},  but we will not state them 
as they are far from optimal.


\subsection{Discussion}
\label{ssec:discussion}
Let us make a couple of comments on the above theorems.

\smallskip

\begin{remark}[\bf Strong coupling] \label{rem-strong}
It has been already mentioned that  there exist potentials in $L^1(\R^2)$ which produce super-linear growth of the counting function
$N(-\Delta-\lambda V)$ as $\lambda\to\infty$, \cite{bl}. Typical examples of such potentials are
\begin{equation} \label{local}
V_\sigma(x) = \left\{
\begin{array}{l@{\quad}cr}
r^{-2}\, |\log r|^{-2}\, \big(\log |\log r|\big)^{-1/\sigma} & \text{if\, \,} & r<
e^{-2}  \\
0 & \, \, \text{if \, }  & r \geq e^{-2} 
\end{array}
\right. \qquad r=|x|,
\end{equation}
and  
\begin{equation} \label{global}
W_\sigma(x) = \left\{
\begin{array}{l@{\quad}cr}
r^{-2}\, (\log r)^{-2}\, (\log \log r)^{-1/\sigma} & \ \ \, \text{if\, \,} & r>
e^{2}  \\
0 & \ \ \ \text{if \, }  & r \leq e^{2} 
\end{array}
\right.  
\quad \qquad r=|x|.
\end{equation}
In particular it follows from \cite[Sec.~6]{bl} that 
\begin{equation} \label{bl-asymp}
\lim_{\lambda\to\infty} \lambda^{-\sigma}\,  N( -\Delta -   \lambda V_\sigma) = \lim_{\lambda\to\infty} \lambda^{-\sigma}\,  N( -\Delta -   \lambda W_\sigma) = \frac{\Gamma\big(\sigma-\frac 12\big)}{2\sqrt{\pi}\ \Gamma(\sigma)}\qquad 	\forall\, \sigma >1.
\end{equation}

It turns out that these effects partially persist even in the presence of a magnetic field. For magnetic Schr\"odinger operators this was proved  in \cite{kov-11}.
For Pauli operators we prove in Section \ref{sec-strong-coupling} 
that, as $\lambda\to\infty$,
\begin{align}
N(\PP  -   \lambda V_\sigma) \, & \asymp\, \lambda^\sigma,  \qquad \forall\ \sigma >1, \qquad  \forall \ \alpha\in\R \label{V_sigma}  \\[5pt]
N(\PP  -   \lambda W_\sigma) \, & \asymp\, \lambda^\sigma,  \qquad \forall\ \sigma >1, \qquad \forall \ \alpha\in\Z. \label{W_sigma} 
\end{align}

Here, for two positive functions $f$ and $g$, we write $f(x) \asymp g(x)$ as $x\to \infty$ if there exist positive constants $K_1 <K_2$ such that 
$$
K_1 \leq \liminf_{x\to \infty} \frac{f(x)}{g(x)} \leq \limsup_{x\to \infty} \frac{f(x)}{g(x)} \leq K_2.
$$

 Let us mention that although the potentials $V_\sigma$ and $W_\sigma$ produce similar behavior of the counting function in the limit $\lambda\to\infty$, their nature is completely different. For $V_\sigma$ is compactly supported and singular in the origin, while $W_\sigma$ is bounded and slowly vanishing at infinity. Notice that neither  $V_\sigma$ nor  $W_\sigma$  belongs to  $L^1(\R^2, \big|\log |x| \big|\, dx)$, but both of them are in $L^1(\R_+, L^p(\Sph))$ for any $p\geq 1$. Indeed, since $V_\sigma$ and $W_\sigma$ are radial,  
\begin{equation}
\|V_\sigma\|_{1,p} =(2\pi)^{\frac 1p -1}\,   \|V_\sigma\|_{L^1(\R^2)} < \infty, \qquad \|W_\sigma\|_{1,p} =(2\pi)^{\frac 1p -1}\,  \|W_\sigma\|_{L^1(\R^2)} < \infty
\end{equation}
for all $\sigma>1$ and all $p\geq 1$. 
Meanwhile, the upper bounds in Theorem \ref{thm-main-pauli} grow linearly in $\lambda$.  In combination with \eqref{V_sigma},  \eqref{W_sigma} and the above equation this shows that the logarithmic weights in  \eqref{clr-pauli-1} and \eqref{clr-pauli-2} cannot be removed. 

However, in the case of non-integer $\alpha$, the logarithmic weight is needed only locally. In fact, since $W_\sigma\in  L^1_{\rm loc}(\R^2, \big|\log |x| \big|\, dx)$, Theorem \ref{thm-main-pauli}  implies that if $\alpha\not\in\Z$, then $W_\sigma$ {\it does not} produce a super-linear growth of the counting function. 
This is compatible with equation \eqref{W_sigma}, or more generally with the hypotheses of Proposition \ref{prop-semiclass-1}.  Same remarks apply  to inequalities  \eqref{clr-schr-1} and \eqref{clr-schr-2}.

\end{remark}

\smallskip

\begin{remark}[\bf Condition $p>1$] \label{rem-p>1} 
Equation \eqref{V_sigma} also implies that the condition $p>1$ in Theorems \ref{thm-main-pauli} and Corollary \ref{cor-main-schr} is sharp, i.e.~the upper bounds \eqref{clr-pauli-1}-\eqref{clr-schr-2} do not hold if $p=1$. To see this, consider the translated potential $V_\sigma(\cdot -x_0)$, with $x_0\neq 0$ and with $\sigma >1$. Then, $\|V_\sigma(\cdot-x_0 ) \|_{1,p} =\infty$ for all $p>1$, whereas
\begin{align*}
\|V_\sigma(\cdot-x_0 ) \|_{1,1} = \|V_\sigma(\cdot-x_0 ) \|_{L^1(\R^2)} = \|V_\sigma  \|_{L^1(\R^2)} < \infty. 
\end{align*}
At the same time, for any  $x_0\neq 0$, we have
\begin{align*}
V_\sigma(\cdot-x_0 )\in L^1(\R^2, \big|\log |x| \big|\, dx).
\end{align*} 
For $p=1$, we would therefore obtain upper bounds on the counting functions that grow linearly in $\lambda$. However, by \eqref{bl-asymp} resp.\ \eqref{V_sigma} and translational invariance, 
$$
N(\PP  -   \lambda V_\sigma(\cdot-x_0 ) ) \, \asymp\, \lambda^\sigma, \qquad N\big( (i\nabla +A)^2 -   \lambda V_\sigma(\cdot-x_0 )\big) \, \asymp\, \lambda^\sigma  \qquad   \forall \ \alpha\in\R\
$$
as $\lambda\to\infty$.

\end{remark}

\smallskip

\begin{remark}[\bf Weak coupling] \label{rem-weak}
The estimates stated in Theorems \ref{thm-main-pauli} and Corollary \ref{cor-main-schr} display the correct behavior also in the weak coupling limit $\lambda\to 0$, cf.~equations \eqref{weak-pauli} 
and \eqref{weak-schr}. 
The presence of the additional factor $m(\alpha)$ in \eqref{clr-pauli-1} and \eqref{clr-pauli-2} is yet another consequence of the spin-orbit coupling which produces exactly $m(\alpha)$ negative eigenvalues of the perturbed Pauli operator in the low energy limit, see \eqref{weak-pauli}.  For the asymptotic expansion of these  eigenvalues we refer to \cite{bcez, fmv, kov, ba}, see also the recent preprint \cite{fi-kr}.
When the the spin-orbit coupling 
is neglected, the weakly coupled eigenvalues disappear, \cite{lw, weidl}. Accordingly,  the  factor $m(\alpha)$ is absent in estimates  \eqref{clr-schr-1} and \eqref{clr-schr-2}.
\end{remark}

\smallskip

\begin{remark}[\bf Long range potentials] \label{rem-long-range}
We have already pointed out in Remark \ref{rem-strong} that the logarithmic weight on the right hand side of \eqref{clr-pauli-2} prevents the application of this estimate to potentials 
which decay as slowly as  $W_\sigma$. In Section \ref{sec-slow} we show that the last term in \eqref{clr-pauli-2} can be replaced by a  different functional of $V$ in such a way that the resulting upper bound covers also potentials of the type $W_\sigma$, see Theorem \ref{pauli-long-range}. 
\end{remark}

\smallskip

\begin{remark}[\bf Condition $\alpha\neq 0$]  \label{rem-alpha-zero}
Corollary \ref{cor-main-schr} follows from the proof of Theorem \ref{thm-main-pauli}, in particular from Proposition \ref{prop-H+} under the hypotheses $\alpha\neq 0$. It is natural to expect that 
estimates  \eqref{clr-schr-1} and \eqref{clr-schr-2} hold even if $\alpha=0$. This question remains open.
\end{remark}

\smallskip


\subsection{Related results}
\label{ssec:literature}
Apart from inequality \eqref{clr-2-dim} for two-dimensional Schr\"odinger operators, results similar to Theorem \ref{thm-main-pauli} were obtained 
in \cite{dflnn} for Hardy-Schr\"odinger operators in dimensions $d\geq 3$. In this case the (unique) weakly coupled eigenvalue arises from subtracting the sharp Hardy 
weight $\frac{(d-2)^2}{4 |x|^2}$ from the Laplace operator. The resulting upper bound on the counting function then include weighted integrals of the potential similar to those 
in estimates  \eqref{clr-pauli-1}-\eqref{clr-schr-2}, see
\cite[Thm.~1]{dflnn}. Fractional  Schr\"odinger operators were discussed in a very recent paper \cite{bfg}. For a weighted version of the Cwikel–Lieb–Rozenblum inequality
for two-dimensional Schr\"odinger operators with Aharonov-Bohm magnetic field we refer to \cite{flr}.

As for estimates on the counting function of the Pauli operator in dimension two, the only existing result is \cite{ffgks},  to the best of our knowledge. The latter
work deals with the Pauli operator on a bounded smooth domain $\Omega$ with magnetic Robin boundary conditions. The authors obtain a sharp {\em lower bound} on the counting function in 
terms of the normalized flux and of the number of boundary components of $\Omega$, see \cite[Thm.~1.1]{ffgks}.

Closely related to the estimates on the counting function are the Lieb-Thirring inequalities, i.e.~
upper bounds on the Riesz means 
\begin{equation} \label{riesz}
\sum_j |E_j|^\gamma = \tr(\PP-V)_\m^\gamma\ , 
\end{equation}
where $E_j$ are the negative eigenvalues of $\PP-V$.  Such inequalities, in dimension two, were obtained in  \cite{es,sob} for all $\gamma \geq 1$,
and in \cite{FK} for all $\gamma>0$ satisfying $\gamma \geq \min\{1, |\alpha|\}$. Note that $N(\PP-V)$ coincides with \eqref{riesz} in the case $\gamma=0$ which is not 
covered by the results of  \cite{es,sob, FK}.

\smallskip

\subsection{Notation} Given a self-adjoint operator $T$ on a Hilbert space $\mathscr{H}$, we indicate the associated counting function with 
$$
N(T)_\mathscr{H}
$$
in those cases where a confusions might arise. In all other cases we drop the subscript $\mathscr{H}$.

Let $X$ be an arbitrary set and $f,g:X \to \mathbb{R}$. In the following, we write 
$$
f(x) \lesssim_{\, \eps} g(x)
$$ 
if
there exists a constant $c_\eps>0$, depending only on $\eps$,  such that $f(x) \geq c_\eps \, g(x)$ for all $x\in X$. Accordingly, $f(x) \lesssim g(x)$ indicates that the implicit constant on the right hand side 
is independent of all the possible parameters introduced in our model.
The symbols $\gtrsim_{\, \varepsilon}$ and $\gtrsim$ are used  similarly. Dependencies on multiple parameters are indicated with multiple subscripts.

We also write $f(x) \asymp g(x) $ if $f(x) \gtrsim g(x)$ and $f(x) \lesssim g(x)$. Note that this is a stronger notion of the symbol $\asymp$ than that introduced in \eqref{V_sigma}, \eqref{W_sigma} but confusion should not arise as we indicate the former notion with the addition ''as $\lambda \to \infty$''.


\subsection{Strategy of the proof}
\label{ssec:strategy}
In this section we briefly sketch the main steps of our proof.
Obviously, 
\begin{equation} \label{sketch-0}
N(\PP  - V )_{L^2(\R^2;\C^2)}  =N( H_\pp -   V )_{L^2(\R^2)} +N( H_\m -   V )_{L^2(\R^2)}\, .
\end{equation}
First, following \cite{weidl, FaKo, FK}  we transform the problem 
to the analysis of  operators $\h_\m$ and $\h_\pp$ acting on weighted $L^2-$spaces, see equation \eqref{passage-weighted}.
One of the main technical tools which we will use in estimating the counting functions of $\h_\ppm- V$ is the following result of Laptev and Netrusov \cite{ln}:

\begin{theorem}[\bf Laptev-Netrusov]\label{thm-ln}
Let $b>0$ and assume that $p>1$. Then there exist $C(b,p)>0$ such that 
\begin{equation} \label{laptev-netrusov}
N\Big( -\Delta +	\frac{b}{|x|^2} - V \Big)_{L^2(\R^2)}  \, \leq\, C(b,p)\,  \|V\|_{1,p} 
\end{equation}
for all $V\in L^1(\R_+, L^p(\Sph))$.
\end{theorem}

 Since $N\big( (i\nabla +A)^2 - B- V \big) =N\big( (i\nabla -A)^2 +B- V\big)$, see equation \eqref{q-forms} below, 
without loss of generality we may and will assume throughout the rest of the paper that 
$$
\alpha\geq 0.
$$
The operator $H_\m$ is then critical, i.e.~it admits weakly coupled negative eigenvalues when perturbed by a negative potential. We introduce the projection operators $P_m$ acting as 
\begin{equation} \label{pm-def}
P_m\, u (r,\theta)= \frac{e^{im\theta}}{2\pi} \int_0^{2\pi} e^{-im\theta'}\, u(r,\theta')\, d\theta'\, \qquad \, m\in\Z\, .
\end{equation}
Clearly, $P_m$ projects $L^2(\R^2)$ onto the subspace of functions with angular momentum $m$. We now set
\begin{equation} \label{P-def}
P=\sum_{m=0}^n P_m  \qquad \text{with} \qquad  n := [\alpha]  ,
\end{equation}
and $P^\perp = 1- P$. Since $\h_\m$ is associated with the quadratic form $\Q_\m$ defined in \eqref{new-qform}, it commutes with $P$. In view of \eqref{passage-weighted} 
and the Cauchy-Schwarz inequality, 
this allows us to estimate the number of negative  eigenvalues of $H_\m-  V$ as follows;
\begin{equation} \label{sketch-0}
N( H_\m-  V) \, \leq\, N\big( \h_\m   -  M_-^2\, V \big) \leq   N\big(P ( \h_\m   -  2 M_-^2 V)P\big)  + N\big(P ^\perp( \h_\m   -  2  M_-^2V)P^\perp\big)\, , 
\end{equation}
where $M_-$ is a positive constant depending only on $B$, see \eqref{M-def}. Iterated application of the Cauchy-Schwarz inequality to the first term on the right side of \eqref{sketch-0} 
further gives 
\begin{equation}  \label{sketch}
N( H_\m-  V) \, \leq\, \left(\sum_{m=0}^n N( h_m^\m -  C_n\, \V) \right) + N\big(P ^\perp( \h_\m   -  2M_-^2  V)P^\perp\big)\, .
\end{equation}
Here $h_m^\m= P_m \h_m\, P_m$ acts in $L^2(\R_+; (1+r)^{-2\alpha}\, r dr)$, $C_n$ is a positive constant,  and 
\begin{equation} \label{V-m-eq}
\V(r) = \frac{1}{2\pi} \int_0^{2\pi} V(r,\theta)\, d\theta\, .
\end{equation}
One of the key ingredients of our proof consists in showing that the operator $P ^\perp\, \h_\m\, P^\perp$ satisfies the Hardy-type bound \eqref{p-perp-lowerb}. Therefore, using Theorem \ref{thm-ln}, we can estimate the 
contribution from the last term in \eqref{sketch} by a constant times $\|V_\pp\|_{1,p}$. This is done in Proposition \ref{prop-Hm-perp-1}. 

As for the first term on the right side of \eqref{sketch}, we note that 
all the operators $h_m^\m$ with $m\in\{0,\dots, n\}$ are critical. However, since they act in $L^2(\R_+; (1+r)^{-2\alpha}\, r dr)$, imposing an additional Dirichlet condition at $r=1$ leads to a rank one perturbation of the resolvent. Hence, by the variational principle, 
\begin{equation}  \label{sketch-2}
 N( h_m^- -  \V) \leq 1+   N( \mathfrak{h}_m^-  -  \V),
\end{equation}
where the operators $\mathfrak{h}_m^-$ act in the same way as $h_m^-$ but with the additional Dirichlet boundary condition $u(1)=0$. Using the Sturm-Liouville theory  in combination with the Birman-Schwinger principle we then estimate the second term in \eqref{sketch-2} by a weighted integral of $\V$. This gives the third term on the right side of \eqref{clr-pauli-1} and \eqref{clr-pauli-2}.
After inserting \eqref{sketch-2} in \eqref{sketch} the additional constant terms add up to $n+1$. See Proposition \ref{prop-H-} for details.

The contribution from $N( \h_\pp -   V )$ in \eqref{sketch-0} is treated in a similar but slightly different way. First, in the case $\alpha=0$, we have $\h_\pp = \h_\m$ and therefore the estimate for $N( \h_\m -   V )$ carries over to $N( \h_\pp -   V )$, resulting in the constant term $m(0)=2$ in \eqref{clr-pauli-1} and \eqref{clr-pauli-2}. Next, we treat the case $\alpha>0$. Here, the operator $\h_\pp $ becomes subcritical. We set $P= P_0$ if $\alpha\not\in\Z$, and $P=P_0+P_\alpha$ if $\alpha\in\Z$. On the range of $P^\perp$ we use the same arguments as above. The counting function of $P (\h_\pp -V) P$ is bounded again with the help of the Sturm-Liouville theory, but this time there is no additional constant term, see Propositions \ref{prop-H+radial+alpha} and \ref{prop-H+}.

\smallskip

Corollary \ref{cor-main-schr} is a consequence of the positivity of the Pauli operator and of  Proposition \ref{prop-H+}. It should be pointed out however, that inequalities \eqref{clr-schr-1} and \eqref{clr-schr-2} were previously known only for radial magnetic fields and for $p=\infty$, cf.~\cite[Sec.~3.2]{kov-11}. The approach of the present paper, which relies on estimating the counting function of the magnetic 
Schr\"odinger operator by the counting function of the subcritical component of the Pauli operator, works even without assuming the axial symmetry of $B$.

\section{\bf  Preliminaries}
From now on we will assume, without loss of generality, that 
\begin{equation} \label{V-positive}
V\geq 0\, .
\end{equation} 
As it is often the case when dealing with the  Pauli operator, we introduce the function
\begin{equation}  \label{h-eq}
h(x) = \frac{1}{2\pi} \int_{\R^2} B(y) \log |x-y|\, dy .
\end{equation}
Standard regularity arguments  imply that under condition \eqref{ass-B} we have $h\in W^{1,\infty}(\R^2)$. Since $-\Delta h = B$ in the sense of distributions, it follows that the vector field  
\begin{equation} \label{A-h-def} 
A_h = (\partial_{x_2} h, - \partial_{x_1} h)
\end{equation}
satisfies  $\nabla\times  A_h = B$ and $|A_h| \in L^\infty(\R^2)$. Owing to the gauge invariance of $N(\PP -V)$ and  $N\big ( (i\nabla +A)^2- V\big )$ we can assume without loss of generality that $|A|\in L^\infty(\R^2)$. We will work in the sequel in the gauge $A_h$. The operators $H_\ppm$ are then associated to the quadratic forms
\begin{equation}\label{q-forms}
Q_\ppm[u] = \int _{\R^2} \big ( |(i\nabla +A_h) u|^2 \pm B |u|^2 \big )\, dx, \qquad u\in H^1(\R^2)\, .
\end{equation}

The functions $h$ defined in \eqref{h-eq}  satisfies
\begin{equation} \label{h-asymp}
h(x) = |x|^{\alpha}\big(1+ \mathcal{O}(|x|^{-1})\big) , \qquad |x|\to \infty.
\end{equation}
Hence  the constants $\mu_\ppm$ and  $m_\ppm$ defined by  
\begin{equation} \label{h-bounds}
\mu_\ppm := \inf_{x\in \R^2} \frac{e^{\pm h(x)}}{(1+|x|)^{\pm \alpha}}\qquad \text{and} \qquad m_\ppm :=  \sup_{x\in \R^2} \frac{e^{\pm h(x)}}{(1+|x|)^{\pm \alpha}}
\end{equation}
depend only on $B$ and satisfy
\begin{equation}  \label{mu-m}
0<\mu_\ppm \leq m_\ppm < \infty.
\end{equation}
Let 
\begin{equation}  \label{M-def}
M_\ppm  = \frac{m_\ppm}{\mu_\ppm} \, \in [1,\infty).
\end{equation}
A  standard calculation, based on the ground-state representation $u= e^{h} v$, gives
\begin{equation} \label{form-factor}
Q_\ppm[e^{h} v] = \int_{\R^2} e^{\pm 2h} \, |(\partial_{x_1} \mp i \partial_{x_2}) v |^2\, dx  .
\end{equation}
Thanks to \eqref{mu-m} we thus conclude that 
\begin{equation} 
Q_\ppm[e^{h} v] \, \geq\,  \mu_\ppm^2\,  \Q_\ppm[v]  \qquad \forall\, v \in D_\ppm(\alpha) ,
\end{equation}
where 
\begin{equation}  \label{new-qform}
 \Q_\ppm[v] =   \int_{\R_+} \int_0^{2\pi}  (1+r)^{\pm 2\alpha}\, |(\partial_r \mp i r^{-1}  \partial_\theta) v |^2\, r\, dr d\theta
\end{equation}
with the form domain
\begin{equation} 
 D_\ppm(\alpha)= \big\{ v\in H^1_{\rm loc}(\R^2):\, \int_{\R^2} (1+|x|)^{\pm 2\alpha}\,(|\nabla v|^2 +|v|^2)\, dx < \infty \big\} .
\end{equation}
Let $\h_\ppm$ denote the operators associated with the quadratic forms $ \Q_\ppm[v]$ on the weighted spaces $L^2(\R^2;(1+|x|)^{\pm 2\alpha}\,dx)$ respectively. From equations \eqref{mu-m} and \eqref{M-def} we deduce that 
\begin{equation}  \label{passage-weighted}
N( \h_\ppm - M_\ppm^{-2}\, V)_{L^2(\R^2;(1+|x|)^{\pm 2\alpha}\,dx )} \, \leq\, N( H_\ppm -V)_{L^2(\R^2)} \, \leq N( \h_\ppm - M_\ppm^2\, V)_{L^2(\R^2;(1+|x|)^{\pm 2\alpha}\,dx )} \, .
\end{equation}

In the sequel we will often use the following elementary  bound. Let $\Pi$ be projection operator on a Hilbert space $\mathscr{H}$, and let $\Pi^\perp=1-\Pi$. If $V\geq 0$, then
the Schwarz inequality implies that for all $\eps>0$ 
\begin{equation} \label{schwarz-eps} 
\Pi^\perp V \Pi +\Pi V \Pi^\perp \leq \eps\, \Pi V \Pi +\eps^{-1} \Pi^\perp V \Pi^\perp
\end{equation} 
in the sense of quadratic forms on $\mathscr{H}$.

\section{\bf Upper bound on $N( H_\m -  V)$ }  
\label{sec-upperb-minus}
The main result of this section is Proposition \ref{prop-H-}. 
In view of equation \eqref{passage-weighted}, it suffices to prove the same upper bound for $N( \h_\m -  V) $. Since $P$ commutes with $\h_\m$, upon setting $\Pi=P$ and $\eps=1$ in \eqref{schwarz-eps} we get
\begin{align}  \label{split-1}
N( \h_\m -  V) &\,  \leq\,  N(P \, \h_\m P -  2PVP)   + N(P ^\perp \h_\m P^\perp -  2P^\perp V P^\perp)\, , 
\end{align}
where all the operators act on the weighted space $L^2(\R^2;(1+|x|)^{- 2\alpha}\,dx ) $. We estimate the terms on the right hand side individually. For the second term we have 

\begin{proposition} \label{prop-Hm-perp-1}
Let $B$ satisfy \eqref{ass-B}.  Assume that $V \in L^1(\R_+, L^p(\Sph))$ for some $p>1$. Then 
there exists a constant $C=C(B,p)$ such that 
\begin{equation}  \label{hm-perp-upperb}
N(P ^\perp\, \h_\m\, P^\perp -  P^\perp\, V\, P^\perp)\, \leq\, C\, \| V \|_{1,p} \, .
\end{equation}
\end{proposition}

\begin{proof}
We shall prove that 
\begin{equation} \label{eq-enough}
\langle v,  P^\perp\, \h_\m\, P^\perp  v \rangle  \,  \gtrsim_{\, \alpha} \,  \int_{\R^2} (1+|x|)^{-2\alpha}\big( |\nabla v|^2+   |x|^{-2} |v|^2\big)\,dx \, , \qquad v\in D_\m(\alpha).
\end{equation}

We will then show that this in turn implies
\begin{equation*} 
\langle \psi , \U^* P^\bot \mathcal H_\m P^\bot \U\,  \psi  \rangle  \, \gtrsim_{\, \alpha}\,  \int_{\R^2} \Big (|\nabla \psi|^2 + \frac{|\psi|^2}{|x|^2} \Big)\,dx \, ,
\end{equation*}
where $\U: L^2(\R^2,dx) \to L^2(\R^2,(1+|x|)^{-2\alpha}dx)$ is a unitary operator. The statement of the proposition then follows upon an application of Theorem \ref{thm-ln}.

By density it suffices to prove the estimate \eqref{eq-enough} for all $v\in C_0^\infty(\R^2)$. Let $\varphi= P^\perp\, v$. First of all, let us note 
that since $\varphi$ is orthogonal to the space of radial functions, the well-known Hardy inequality, see e.g.~\cite{bl}, implies 
\begin{align} \label{hardy-orth}
\int_{\R^2}  (1+|x|)^{-2\alpha} |\nabla \varphi |^2\,dx
& \geq  \int_0^\infty (1+r)^{-2\alpha} \, r \int_0^{2\pi}  | r^{-1} \partial_\theta \varphi|^2 \, d\theta\, dr \,  \geq \, \int_{\R^2} (1+|x|)^{-2\alpha}   |x|^{-2} |\varphi|^2\,dx \,.
\end{align}
Now let
\begin{equation}  \label{fi-m}
\varphi_m(r) = \frac{1}{2\pi} \int_0^{2\pi}\!\! e^{-im\theta}\, \varphi(r,\theta)\, d\theta 
\end{equation}
denote the Fourier coefficients of $\varphi$. 
By a direct calculation, similar to the one in \cite[Sec.~10]{weidl}, we then get 
\begin{equation} \label{h-fourier}
\begin{aligned}  
 \langle \varphi,   \h_\m \varphi \rangle = \Q_-[\varphi] & = \sum_{m\in\Z} \int_0^\infty (1+r)^{-2\alpha}  \big | \varphi'_m(r) - \frac{m\, \varphi_m(r)}{r} \big|^2\, r\, dr  \\[5pt]
&=   \sum_{m\in\Z}   \int_0^\infty (1+r)^{-2\alpha}\,    r^{1+2m}\, \big | \partial_r( r^{-m}\, \varphi _m(r))\big|^2\, dr \, .
\end{aligned}
\end{equation} 
Let us consider the integrals in the sum on the right hand side. 

First, suppose that $m<0$. After an integration by parts, we get
\begin{align} \label{m<0}
\int_0^\infty (1+r)^{-2\alpha}\,   \big | \varphi_m' - \frac{m\, \varphi_m}{r} \big|^2\, r\, dr & = \int_0^\infty (1+r)^{-2\alpha}\,   \Big( |\varphi'_m|^2 + \frac{m^2\, |\varphi_m|^2}{r^2} -2\alpha m \frac{ |\varphi_m|^2}{(1+r)r} \Big)\, r\, dr
\nonumber  \\[5pt]
& \geq \int_0^\infty (1+r)^{-2\alpha}\,   \Big( |\varphi'_m|^2 + \frac{m^2\, |\varphi_m|^2}{r^2} \Big)\, r\, dr.
\end{align}

If $m>\alpha$, we let $g_m= r^{-m} \varphi _m$. Then $\liminf_{t\to\infty} |g_m(t)| =0$ and  Theorem \ref{thm-class-1} applied with $U(t) = t^{2m +1} (1+t)^{-2\alpha}$ 
and $W(t) = t^{-2}\, U(t)$ gives 
\begin{equation}
\int_0^\infty (1+r)^{-2\alpha}\,  r^{2m+1}\, |g'_m(r)|^2\, dr\, \gtrsim_{\, \alpha} \, (m-\alpha)^2  
\int_0^\infty (1+r)^{-2\alpha}\,  r^{2m-1}\, |g_m(r)|^2\, dr .
\end{equation}
Since  $(m-\alpha)^2  \gtrsim_{\, \alpha}  m^2$ for all $m\geq n+1$, this implies
\begin{equation}
\int_0^\infty  (1+r)^{-2\alpha}  \big | \varphi'_m(r) - \frac{m\, \varphi_m(r)}{r} \big|^2\, r\, dr \,  \gtrsim_{\, \alpha} \,  m^2  
\int_0^\infty (1+r)^{-2\alpha}\,   \frac{|\varphi_m(r)|^2}{r^2} \, r\, dr .
\end{equation}
Using the inequality
\begin{equation} \label{young}
(a+b)^2 + \eps^2 \, b^2 \,  \gtrsim_{\, \eps} a^2 + b^2,
\end{equation}
we deduce that 
\begin{equation} \label{m>n}
\int_0^\infty  (1+r)^{-2\alpha}  \big | \varphi'_m(r) - \frac{m\, \varphi_m(r)}{r} \big|^2\, r\, dr \,  \gtrsim_{\, \alpha} \, 
\int_0^\infty (1+r)^{-2\alpha}\, \Big(  |\varphi'_m(r) |^2+ \frac{m^2 |\varphi_m(r)|^2}{r^2}\Big) \, r\, dr\,  . 
\end{equation}

Since, by definition of $P^\perp$,
\begin{equation}  \label{fi-m-zeros}
\varphi_m = 0 \qquad \forall \, m\in \{0,1,\dots, n\},
\end{equation} 
equations \eqref{m>n} and \eqref{m<0} combined with the Parseval identity imply
\begin{equation}  \label{Q-lowerb1}
 \Q_-[\varphi]  \, \gtrsim_{\, \alpha} \,  \sum_{m\in\Z} \int_0^\infty (1+r)^{-2\alpha}   \Big( |\varphi'_m|^2 + \frac{m^2\, |\varphi_m|^2}{r^2} \Big)\, r\, dr  =  \int_{\R^2}  (1+|x|)^{-2\alpha}\,  |\nabla \varphi |^2\,dx\, .
 \end{equation}
Note that using the obvious upper bound $(a+b)^2 \leq 2 a^2 + 2b^2$, we could replace here $\gtrsim_{\, \alpha}$ by $\asymp_{\,\alpha}$. In view of \eqref{hardy-orth}, the estimate \eqref{Q-lowerb1} proves \eqref{eq-enough}. 

To proceed, we consider the unitary operator $\U: L^2(\R^2,dx) \to L^2(\R^2,(1+|x|)^{-2\alpha}dx)$ given by  
$\U \psi=  (1+|x|)^{\alpha} \psi$. Note that  $\U$ commutes with $P^\perp$. Hence by \eqref{Q-lowerb1} and \eqref{hardy-orth},
\begin{equation} \label{no-weight}
\langle \psi , \U^* P^\bot \mathcal H_\m P^\bot \U\,  \psi  \rangle  \, \gtrsim_{\, \alpha}\,  \int_{\R^2} (1+|x|)^{-2\alpha} |\nabla ((1+|x|)^\alpha \psi)|^2\,dx \, \geq \, \int_{\R^2} \frac{|\psi|^2}{|x|^2} \,dx \, .
\end{equation} 
Meanwhile, integration by parts shows that
\begin{align*}
		\int_{\R^2} (1+|x|)^{-2\alpha} |\nabla ((1+|x|)^\alpha \psi)|^2\,dx
		& = \int_{\R^2} (|\nabla \psi|^2 - \alpha (1+|x|)^{-2} (|x|^{-1}-\alpha ) |\psi|^2 )\,dx \,.
	\end{align*}
Combining this with \eqref{no-weight} we find, for any $\eps\in[0,1]$,
\begin{align*}
		\int_{\R^2} (1+|x|)^{-2\alpha} |\nabla ((1+|x|)^\alpha \psi)|^2\,dx & \geq\,  \eps \int_{\R^2} |\nabla \psi|^2 \,dx  \\[4pt]
		&\quad + \int_{\R^2} ((1-\eps)|x|^{-2} - \eps \alpha (1+|x|)^{-2} (|x|^{-1}-\alpha )) |\psi|^2 \,dx \,.
\end{align*}
As in \cite{FK} it follows that upon setting 	
$$
\eps = \big[ \sup_{r>0} ( 1+ \alpha (1+r)^{-2} r(1-\alpha r) ) \big]^{-1}
$$
we have
$$
(1-\eps)|x|^{-2} - \eps \alpha (1+|x|)^{-2} (|x|^{-1}-\alpha )\geq 0
	\qquad \text{for all}\ x\in\R^2 \,.
$$
Altogether we thus get
\begin{equation} \label{p-perp-lowerb}
\langle \psi , \U^* P^\bot \mathcal H_\m P^\bot \U\,  \psi  \rangle  \, \gtrsim_{\, \alpha}\,  \int_{\R^2} \Big (|\nabla \psi|^2 + \frac{|\psi|^2}{|x|^2} \Big)\,dx \, .
\end{equation}
Hence there exists a constant $c_\alpha>0$ such that 
\begin{align*}
 N(P ^\perp \h_\m P^\perp - P^\perp VP^\perp)_{L^2(\R^2;(1+|x|)^{- 2\alpha}\,dx )} \, & \leq \,  N(P ^\perp ( -\Delta + |x|^{-2} -c_\alpha V)P^\perp)_{L^2(\R^2)}  \\[5pt]
 & \leq N( -\Delta + |x|^{-2} -c_\alpha V)_{L^2(\R^2)} \, .
\end{align*}
Inequality \eqref{hm-perp-upperb} now follows from Theorem \ref{thm-ln}.
\end{proof}

Next we consider the first term on the right hand side of \eqref{split-1}. Since $\h_\m$ commutes with $P_m$, equation \eqref{V-positive} and the Schwarz inequality imply 
\begin{equation} \label{split-n}
P\, \h_\m\, P-  P\, V\, P  \geq  \sum_{m=0}^n  P_m \h_\m P_m - c_n \sum_{m=0}^n P_m V P_m, \qquad c_n = 1+\frac{n(n+1)}{2}\, .
\end{equation}
Now, for any $u\in D_-(\alpha)$ we have 
\begin{equation} 
(P_m V P_m u)(r,\theta) = \frac{e^{im\theta}}{2\pi}\, u_m(r)\, \V(r),
\end{equation}
where 
\begin{equation*} 
u_m(r) = \frac{1}{2\pi}\int_0^{2\pi} e^{-im\theta} u(r,\theta)\, d\theta\, 
\end{equation*}
is the $m$-th Fourier coefficient of $u$.
Combined with inequality \eqref{split-n}, this gives 
\begin{equation} \label{upperb-PHP}
N(P \, \h_\m P -  PVP)_{L^2(\R^2;(1+|x|)^{- 2\alpha}\,dx )}  \, \leq\, \sum_{m=0}^n N( h_m^\m -  c_n \V)_{L^2(\R_+; (1+r)^{-2\alpha}\, r dr)}
\end{equation} 
Here, $h_m^-$ is the operator associated in $L^2(\R_+; (1+r)^{-2\alpha}\, r dr)$ with the quadratic from
\begin{equation} \label{qm-form}
  \int_0^\infty (1+r)^{-2\alpha}  \big | v'(r) - \frac{m v}{r} \big|^2\, r\, dr  =  \int_0^\infty (1+r)^{-2\alpha}\,    r^{1+2m}\, \big | \partial_r( r^{-m}\, v(r))\big|^2\, dr 
\end{equation} 
on the form domain
\begin{equation} 
\big\{  v\in H^1_{\rm loc}(\R_+):\, \int_0^\infty (1+r)^{- 2\alpha}\,(| v'|^2 +|v|^2)\, r\, dr < \infty \big\}.
\end{equation}
Now we impose a Dirichlet condition at $r=1$. Since this is a rank one perturbation of the resolvent, we conclude that 
\begin{align} \label{DBC-1}
0\leq N( h_m^- -  \V)_{L^2(\R_+; (1+r)^{-2\alpha}\, r dr)} -  N( \mathfrak{h}_m^- -  \V)_{L^2((0,1); (1+r)^{-2\alpha}\, r dr)}    -  N( \mathfrak{h}_m^- -  \V)_{L^2((1,\infty); (1+r)^{-2\alpha}\, r dr)}  &\leq 1,
\end{align} 
where the operator $\mathfrak{h}_m^-$ acts in $L^2((0,1); (1+r)^{-2\alpha}\, r dr)$ respectively $L^2((1,\infty); (1+r)^{-2\alpha}\, r dr)$ as $h_m^-$ with the additional Dirichlet boundary condition at $r=1$. Replacing the integral weight $ (1+r)^{-2\alpha}$ by $1$ on $(0,1)$ and by $r^{-2\alpha}$ on $(1,\infty)$, we find that 
\begin{equation} \label{reduced-weights}
\begin{aligned} 
N( h_{m,1}^- -  4^{-\alpha}\, \V)_{L^2((0,1); \, r dr)} \, \leq\, N( \mathfrak{h}_m^- -  \V)_{L^2((0,1); (1+r)^{-2\alpha}\, r dr)}  \,& \leq \, N( h_{m,1}^- -  4^\alpha\, \V)_{L^2((0,1); \, r dr)}  \\[8pt]
N( h_{m,2}^- -  4^{-\alpha}\, \V)_{L^2((1,\infty); \, r^{1-2\alpha} dr)} \leq N( \mathfrak{h}_m^- -  \V)_{L^2((1,\infty); (1+r)^{-2\alpha}\, r dr)}  \, &\leq \, N( h_{m,2}^- -  4^\alpha\, \V)_{L^2((1,\infty); \, r^{1-2\alpha} dr)}
\end{aligned}
\end{equation} 
where the operators $h_{m,1}^-$ and $h_{m,2}^-$ are associated with quadratic forms 
\begin{equation} \label{q-forms-dirichlet}
\begin{aligned} 
q_{m,1}^-[v]  &=  \int_0^1   \big | \partial_r( r^{-m}\, v(r))\big|^2\, r^{1+2m}\, dr ,  \qquad\ \ v \in H^1((0,1), r dr), \qquad  \ \quad v(1)=0
\\[4pt]
q_{m,2}^-[v]  &=  \int_1^\infty   \big | \partial_r( r^{-m}\, v(r))\big|^2\, r^{1+2m-2\alpha}\, dr ,  \quad v\in H^1((1,\infty), r^{1-2\alpha} dr), \  \ v(1)=0.
\end{aligned}
\end{equation}

It remains to estimate the counting functions of the operators $h_{m,1}^-$ and $h_{m,2}^-$ which appear on the right side of    \eqref{reduced-weights}. The next lemma provides an upper bound for 
$N( h_{m,1}^- -   \V)_{L^2((0,1); \, r dr)}$.

\begin{lemma}  \label{lem-01}
Let $0\leq \V\in L^1((0,1); |\log r| r dr)$ and $m\in\Z$. Then 
\begin{align} 
 N( h_{0,1}^- -    \V)_{L^2((0,1); \, r dr)} &  \, \lesssim\, \int_0^1  \V(r)\, |\log r|\, r\, dr\, \label{m=0} \\[5pt]
  N( h_{m,1}^- -    \V)_{L^2((0,1); \, r dr)} &  \, \lesssim_{\, m}\, \int_0^1  \V(r)\,  r\, dr\, \qquad\qquad \forall\ m\neq 0 \label{m-neq-0}.
\end{align}
\end{lemma}

\begin{proof} 
Let $v\in C_0^\infty(0,1)$. 
An application of Corollary \ref{cor-class-1} with $f(r) = r^{-m}\, v(r), \, U(r)= r^{1+2m}$ and $W(r) =r^{-2}\, U(r)$ gives
$$
q_{m,1}^-[v] \, \gtrsim_{\, m}\,  \int_0^1 |v|^2\, r^{-1}\, dr\, .
$$
Since
$$
q_{m,1}^-[v] = \int_0^1  \big | v'(r) - \frac{m\, v}{r} \big|^2\, r\, dr,
$$
we deduce from \eqref{m>n} that 
\begin{equation} 
q_{m,1}^-[v] \, \gtrsim_{\, m}\,  \int_0^1 |v'|^2\, r\, dr\, .
\end{equation}
It follows that there exists a constant a constant $c_m$ such that 
\begin{equation}  \label{hm1-upperb}
 N( h_{m,1}^- -    \V)_{L^2((0,1); \, r dr)}  \, \leq \,  N\big ( -r^{-1} \partial_r r \partial_r +\frac{m^2}{r^2}- c_m   \V\big)_{L^2((0,1); \, r dr)} \, .
\end{equation} 
Moreover, the unitary mapping $\U: L^2((0,1); r dr) \to L^2((0,1),  dr)$ given by $\U v= \sqrt{r}\, v$ shows that  
\begin{equation}  \label{unitary-transf}
N\big ( -r^{-1} \partial_r r \partial_r +\frac{m^2}{r^2}- c_m   \V\big)_{L^2((0,1); \, r dr)} =   N\big (   -\partial_r^2 - \frac{1}{4r^2}  +\frac{m^2}{r^2} -   \V\big)_{L^2((0,1);  dr)}\, ,
\end{equation} 
where the operator on the right hand side is subject to Dirichlet boundary conditions at $r=0$ and $r=1$. Since the operator $ -\partial_r^2 - \frac{1}{4r^2}$ coincides with the 
radial part of the two-dimensional Laplacian restricted to functions which vanish for $|x|\geq 1$,  the upper bound \eqref{m-neq-0}  follows from 
\cite[Thm.~1.2]{ln} and equations \eqref{hm1-upperb} and \eqref{unitary-transf}. 

If  $m=0$, then a standard calculation, see e.g.~\cite[Sec.~2.3]{ls},  shows that 
\begin{equation}
\Big(-\partial_r^2 -\frac{1}{4r^2} +\kappa^2\Big)^{-1}(r,r')= \sqrt{r r'}\, 
I_0(\kappa r) \big[  K_0(\kappa r') - \beta_\kappa I_0(\kappa r')]   \qquad  0 < r\leq r'\leq  1,
\end{equation}
where $K_\nu$ and $I_\nu$ denote the modified  Bessel functions, see \cite[Sec.~9.6]{as}, and where
\begin{equation} \label{beta-kappa}
\beta_\kappa = \frac{K_0(\kappa)}{I_0(\kappa)}\, .
\end{equation}
Note that, since $K_\nu$ is decreasing and $I_\nu$ is increasing, $K_0(\kappa r') - \beta_\kappa I_0(\kappa r') \geq 0$ for all $0 < r'\leq  1$. 
The Birman-Schwinger operator has the integral kernel 
\begin{equation} 
\Big(\sqrt{ \V}\big (-\partial_r^2 -\frac{1}{4r^2} +\kappa^2 \big)^{-1} \sqrt{ \V}\Big) (r,r' ) = \sqrt{ \V(r)}  \sqrt{r r'}\, 
I_0(\kappa r) \big[  K_0(\kappa r') - \beta_\kappa I_0(\kappa r')]  \sqrt{ \V(r')}\, .
\end{equation}
From the asymptotic expansions of Bessel functions:
\begin{equation}  \label{KI-zero}
	\begin{aligned}
		K_\nu(z) & = 
		\begin{cases}
			-\log z +  C + \mathcal O(z^2 |\log z|) & \text{if}\ \nu=0 \,,\\[3pt]
			\left( \frac z2 \right)^{-\nu} \tfrac12\, \Gamma(\nu) + \mathcal O(z^{\min\{\nu, 2-\nu\}})& \text{if}\ \nu \neq 1 \,, \\[3pt]
			z^{-1} + \mathcal O(z|\log z|) & \text{if}\ \nu=1 \,, \\			
		\end{cases}
		\qquad\text{as}\ z\to 0 \,, \\[6pt]
		I_\nu(z) & = \left( \frac z2 \right)^\nu \Gamma(\nu+1)^{-1} + \mathcal O(z^{2+\nu}) \qquad\text{as}\ z\to 0 \,,
	\end{aligned}
\end{equation} 
see  \cite[Eqs.~9.6.12, 9.6.13]{as},  we deduce that 
\begin{equation} \label{kernel-K}
\lim_{\kappa\to 0} \Big(\sqrt{ \V}\big (-\partial_r^2 -\frac{1}{4r^2} +\kappa^2 \big)^{-1} \sqrt{ \V}\Big) (r,r' )  = - \sqrt{\V(r)}\ \sqrt{r r'}\,  \log(\max\{ r, r'\}) \,  \sqrt{\V(r')}
\end{equation}
for all $r,r' \in (0,1)$.  This kernel is positive definite on $(0,1)\times (0,1)$, see Lemma \ref{lem-positive-def_lem3-2}. Moreover, since 
$$
\int_0^1 \Big(\sqrt{\V}\big (-\partial_r^2 -\frac{1}{4r^2} \big)^{-1} \sqrt{\V}\Big) (r,r)\, dr < \infty, 
$$
by assumption on $\V$, it follows from \cite[Thm.~2.12]{si-tr} that the operator $K$ with integral kernel \eqref{kernel-K} is trace class in $L^2(0,1)$. Let us denote its eigenvalues by $\{\mu_j\}_{j\in\N}$. 
The Birman-Schwinger principle  then implies 
\begin{equation}\label{BS-op-trace}
\begin{aligned}
N\big (   -\partial_r^2 -\frac{1}{4r^2} -  \V\big)_{L^2((0,1);  dr)}\, & = \sum_{j: \, \mu_j \geq 1} 1\,  \leq\,   \sum_{j: \, \mu_j \geq 1}  \mu_j  \, \leq \,  \sum_{j\in\N} \mu_j = \tr (K) = \int_0^1 K(r,r)\, dr \\[4pt]
& = \int_0^1 \V(r)\, |\log r|\, r\, dr\, .
\end{aligned}
\end{equation}
Hence equation \eqref{m=0} follows by \eqref{hm1-upperb} and \eqref{unitary-transf}. 
\end{proof}

Now we  estimate $N( h_{m,2}^- -   \V)$, the counting function that appears in the second line of \eqref{reduced-weights}. We distinguish the cases $m<\alpha$ and $m=\alpha$. In the first case, we have

\begin{lemma}  \label{lem-1infty}
Let $0\leq \V \in L^1(\R_+; r dr)$, $\alpha >0$ and $m\in\mathbb{Z}$, $m< \alpha$. Then 
$$
 N( h_{m,2}^- -   \V)_{L^2((1,\infty); \, r^{1-2\alpha} dr)}  \, \lesssim_{\, \alpha, m}\, \int_0^\infty \V(r)\,  r\, dr\, . 
$$
\end{lemma}

\begin{proof}
We mimic the proof of Lemma \ref{lem-01}.
Let $v\in C_0^\infty(1,\infty)$. 
An application of Theorem \ref{thm-class-2}  with $f(r) = r^{-m}\, v(r), \, U(r)= r^{1+2m-2\alpha}$ and $W(r) =r^{-2}\, U(r)$ gives
$$
q_{m,2}^-[v] \, \gtrsim_{\, m}\,  \int_1^\infty |v|^2\, r^{-1-2\alpha}\, dr\, .
$$

Since
$$
q_{m,2}^-[v] = \int_1^\infty  \big | v'(r) - \frac{m\, v}{r} \big|^2\, r^{1-2\alpha}\, dr,
$$
we deduce from \eqref{young} that 
\begin{equation} 
q_{m,2}^-[v] \, \gtrsim_{\, m}\,  \int_1^\infty |v'|^2\, r^{1-2\alpha}\, dr\, .
\end{equation}
Hence there exists a constant $c'_m$ such that 
\begin{equation}  \label{hm2-upperb}
 N( h_{m,2}^- -   \V)_{L^2((1,\infty); \, r^{1-2\alpha} dr)}  \, \leq \,  N\big ( -r^{2\alpha-1}\partial_r\,  r^{-2\alpha+1} \partial_r- c'_m  \V\big)_{L^2((1,\infty); \, r^{1-2\alpha} dr)} \, .
\end{equation} 
By the variational principle, 
\begin{equation} \label{N-1}
 N( -r^{2\alpha-1}\partial_r\,  r^{-2\alpha+1} \partial_r -   \V)_{L^2((1,\infty); \, r^{1-2\alpha} dr)}  \, \leq\,  N( -r^{2\alpha-1}\partial_r\,  r^{-2\alpha+1} \partial_r -   \V)_{L^2(\R_+; \, r^{1-2\alpha} dr)},
\end{equation}
where the operator on the right hand side is subject to Dirichlet boundary condition at $r=0$. Meanwhile, using the mapping $v\mapsto r^{\frac 12 -\alpha} \, v$, 
which maps $L^2(\R_+; \, r^{1-2\alpha} dr)$ unitarily onto $L^2(\R_+;  dr)$, we infer that  
\begin{equation}\label{N-2}
N( -r^{2\alpha-1}\partial_r\,  r^{-2\alpha+1} \partial_r -   \V)_{L^2(\R_+; \, r^{1-2\alpha} dr)}  =   N\big (   -\partial_r^2 + (\alpha^2-1/4) r^{-2} -  \V\big)_{L^2(\R_+;  dr)}\, .
\end{equation} 
Similarly as in the proof of Lemma \ref{lem-01} we thus obtain
\begin{equation}\label{green-0-infty}
\Big(-\partial_r^2 + (\alpha^2-1/4 )  \frac{1}{ r^{2}} +\kappa^2\Big)^{-1}(r,r')= \sqrt{r r'}\ 
I_\alpha(\kappa r)\,  K_\alpha(\kappa r')   \qquad  0 < r\leq r' <  \infty. 
\end{equation}
Since 
\begin{equation}  \label{IK-gamma}
\lim_{\kappa\to 0} I_\alpha(\kappa r)\,  K_\alpha(\kappa r')   = \frac{\Gamma(\alpha)}{2 \Gamma(1+\alpha)} \left( \frac{r}{r'}\right)^\alpha = \frac{1}{2\alpha} \left( \frac{r}{r'}\right)^\alpha,
\end{equation}
 see  \eqref{KI-zero}, we conclude that
 \begin{equation} 
\lim_{\kappa\to 0} \Big(\sqrt{ \V}\big (-\partial_r^2 + (\alpha^2-1/4 )  \frac{1}{ r^{2}} +\kappa^2 \big)^{-1} \sqrt{ \V}\Big) (r,r' ) = \sqrt{ \V(r)}  \frac{\sqrt{r r'}}{2\alpha} \left(\min \Big\{ \frac{r}{r'}, \frac{r'}{r} \Big\}\right)^\alpha   \sqrt{ \V(r')}\, .
\end{equation} 
 
 This kernel is positive definite on $\R_+ \times \R_+$, see Lemma \ref{lem-positive-def_lem3-3}. As in \eqref{BS-op-trace} it follows from the Birman-Schwinger principle that 
\begin{align*}
N\big (   -\partial_r^2 +(\alpha^2-1/4 )  \frac{1}{ r^{2}} -  \V\big)_{L^2(\R_+;  dr)}\, & \leq  \,  \tr \Big (\sqrt{\V}\big (-\partial_r^2 + (\alpha^2-1/4) \frac{1}{ r^{2}}  \big)^{-1} \sqrt{\V}\, \Big ) = \frac{1}{2\alpha}  \int_0^\infty \V(r)\,  r\, dr\, .
\end{align*}
This in combination with \eqref{N-1} and \eqref{N-2} completes the proof.
\end{proof}

In the case $m= \alpha$ we find

\begin{lemma}  \label{lem-1infty-bis}
Let $0 \leq \alpha\in\Z$, and suppose that $0\leq \V\in L^1((1,\infty); (\log r )\, r \, dr)$. Then 
$$
 N( h_{\alpha,2}^- -   \V)_{L^2((1,\infty); \, r^{1-2\alpha} dr)}  \, \leq\, \int_1^\infty \V(r)\, (\log r)\, r \,  dr\, .
$$
\end{lemma}

\begin{proof}
Let $v\in C_0^\infty(1,\infty)$. 
Integration by parts gives 
\begin{align*}
q_{\alpha,2}^-[v]  & = \int_1^\infty     \Big | v' - \frac{\alpha v}{r} \Big|^2\, r^{1-2\alpha}\, dr  = \int_1^\infty  \Big( |v'|^2 - \frac{\alpha^2\, |v|^2}{r^2} \Big)\,   r^{1-2\alpha}\, dr\, .
\end{align*}
As in the proof of Lemma \ref{lem-1infty} we apply the mapping $v\mapsto r^{\frac 12 -\alpha} \, v$ and deduce that
\begin{equation} \label{h-alpha-upperb}
N( h_{\alpha,2}^-  -  \V)_{L^2((1,\infty); \, r^{1-2\alpha} dr)}  =   N\big (   -\partial_r^2 - \frac{1}{4r^2}  -  \V\big)_{L^2((1,\infty);  dr)}\, .
\end{equation} 
Keeping in mind that the operator on the right hands side is subject to Dirichlet boundary condition at $r=1$, we 
calculate the integral kernel of the resolvent using again the Strum-Liouville theory. This gives
\begin{equation}\label{green-1-infty}
\big(-\partial_r^2 - \frac{1}{4r^2}  +\kappa^2\big)^{-1}(r,r')= \sqrt{r r'}\ 
\big [ I_0(\kappa r) -\beta^{-1}_\kappa \,  K_0(\kappa r)  \big ] K_0(\kappa r')   \qquad  1 < r\leq r' <  \infty, 
\end{equation}
with $\beta_\kappa$ given by \eqref{beta-kappa}. Note that $ I_0(\kappa r) -\beta^{-1}_\kappa \,  K_0(\kappa r)  \geq 0$ for all $\kappa>0$ and all $1\leq r $, since $K_\nu$ is decreasing, $I_\nu$ is increasing, and $ I_0(\kappa ) -\beta^{-1}_\kappa \,  K_0(\kappa ) =0$.  With the help of \eqref{KI-zero}  we then get 
$$
\lim_{\kappa\to 0} \Big(  \sqrt{\V} \big(-\partial_r^2 - \frac{1}{4r^2}  +\kappa^2\big)^{-1}  \sqrt{\V} \Big) (r,r') =  \sqrt{\V(r)}\sqrt{r r'}\ \log (\min\{ r, r'\}) \sqrt{\V(r')} \, .
$$
This kernel is positive definite on $(1,\infty) \times (1,\infty)$. We omit the proof as this can be proven similarly to Lemma \ref{lem-positive-def_lem3-2}.
As above, we then get
\begin{align*}
N\big (   -\partial_r^2 - \frac{1}{4r^2}  -  \V\big)_{L^2((1,\infty);  dr)}\, & \leq  \,  \tr \big ( \sqrt{\V}\big (-\partial_r^2 -\frac{1}{4r^2}    \big)^{-1} \sqrt{\V}\, \big )
=   \int_1^\infty \V(r)\,  (\log r)\, r \, dr\, .
\end{align*}
The claim now follows from equation \eqref{h-alpha-upperb}.
\end{proof}

Combining the previous three lemmas yields

\begin{proposition} \label{prop-H-1}
Let $B$ satisfy \eqref{ass-B}. 
\begin{enumerate} 
\item Let $0 <\alpha\notin\Z$. Assume that $V \in L_{\rm loc}^1(\R^2)$ and that  $V  \log |x  | \in L^1(\B_1)$. Then 
there exists a constant $C=C(B)$ such that 
\begin{equation} 
N(P \, \h_\m\, P-  P\, V\, P) \, \leq\, [\alpha] +1 + C \int_{\R^2}  V(x) \big(1+ \id_{\B_1}(x) | \log |x || \big)\, dx .
\end{equation}
\item Let $0 \leq \alpha\in\Z$. Assume that $V \in L^1(\R^2)$ and that  $V  \log |x  | \in L^1(\R^2)$. Then 
there exists a constant $C=C(B)$ such that 
\begin{equation} 
N(P \, \h_\m\, P-  P\, V\, P) \, \leq\, \alpha +1 + C \int_{\R^2}  V(x) \big(1+  | \log |x || \big)\, dx .
\end{equation}
\end{enumerate}
\end{proposition}

\begin{proof}
(1) In view of \eqref{V-m-eq} the result follows by combining equations \eqref{upperb-PHP}, \eqref{DBC-1} and \eqref{reduced-weights}  with Lemmas \ref{lem-01} and \ref{lem-1infty}. (2) Similarly as in (1), the result follows by combining equations \eqref{upperb-PHP}, \eqref{DBC-1} and \eqref{reduced-weights}  with Lemmas \ref{lem-01}, \ref{lem-1infty} and \ref{lem-1infty-bis}.
\end{proof}

We can now state the main result of this section.

\begin{proposition}\label{prop-H-}
Let $B$ satisfy \eqref{ass-B}.  Then we have
\begin{enumerate} 
\item Assume that $0<\alpha \not\in\Z$. Then for any $p>1$ there exist constants $C_1=C_1(B,p)$ and $C_2=C_2(B)$ such that 
\begin{equation} 
N(  H_\m\, -   V) \, \leq\, m(\alpha) + C_1\, \|V\|_{1,p} +C_2 \| V  \log |x|\|_{L^1(\!\B_1)} 
\end{equation}
for all  $V \in L^1(\R_+, L^p(\Sph))$ with  $V  \log |\cdot | \in L^1(\B_1)$.

\item Assume that $0\leq \alpha \in\Z$. Then for any $p>1$ there exist constants $\CC_1=\CC_1(B,p)$ and $\CC_2=\CC_2(B)$ such that 
\begin{equation} 
N(  H_\m\, -   V) \, \leq\, m(\alpha) + \CC_1\, \|V\|_{1,p} +\CC_2 \| V\log |x|\|_{L^1(\R^2)} 
\end{equation}
for all  $V \in L^1(\R_+, L^p(\Sph))$ with  $V  \log |\cdot | \in L^1(\R^2)$.
\end{enumerate}
\end{proposition}

\begin{proof}
By H\"older inequality,  
$$
\|V\|_{L^1(\R^2)}  \leq (2\pi)^{\frac{p-1}{p}}\, \|V\|_{1,p} \qquad p\geq 1,
$$
Hence
the claim follows from equations \eqref{passage-weighted}, \eqref{split-1} and Propositions \ref{prop-Hm-perp-1} and \ref{prop-H-1}.
\end{proof}

\section{\bf Upper bound on $N( H_\pp -  V)$ }  
\label{sec-upperb-plus}

The goal of this section is to find an upper bound on $N( H_\pp -  V)$.  Note that for $\alpha=0$ we have $\h_\pp = \h_\m$. Hence one can conclude from \eqref{passage-weighted}
\begin{align*}
N( H_\pp -V)_{L^2(\R^2)}  \leq N( \h_\pp - M_\pp^2\, V)_{L^2(\R^2)} = N( \h_\m - M_\pp^2\, V)_{L^2(\R^2)}.
\end{align*}
Using Proposition \ref{prop-Hm-perp-1} and Proposition \ref{prop-H-1} (2), we see that the upper bound given in Proposition \ref{prop-H-} (2) also holds for $N( H_\pp -V)$. The statement of Theorem \ref{thm-main-pauli} (2) for $\alpha=0$ now follows from the bounds on $N( H_\pm -V) $. 

\smallskip
In the sequel we will therefore assume that  $\alpha>0$. Our aim is to prove Proposition \ref{prop-H+}. 
We estimate  the counting function of $ \h_\pp -  V$ as follows: if $\alpha\not\in\Z$, we write
\begin{align}  \label{split-2a}
N( \h_\pp -  V)_{L^2(\R^2;(1+|x|)^{2\alpha}\,dx )} & \leq N(P _0\, \h_\pp P_0 -  2P_0VP_0)_{L^2(\R^2;(1+|x|)^{ 2\alpha}\,dx )}  \nonumber \\[4pt]
& \quad + N(P _0^\perp \h_\pp P_0^\perp -  2P_0^\perp V P_0^\perp)_{L^2(\R^2;(1+|x|)^{ 2\alpha}\,dx )} \, ,
\end{align}
and if $\alpha\in\Z$, then
\begin{align}  \label{split-2b}
N( \h_\pp -  V)_{L^2(\R^2;(1+|x|)^{2\alpha}\,dx )} & \leq N(P _0\, \h_\pp P_0 -  4 P_0VP_0)_{L^2(\R^2;(1+|x|)^{ 2\alpha}\,dx )}  \nonumber \\[4pt]
&\quad +N(P _\alpha\, \h_\pp P_\alpha -  4P_\alpha VP_\alpha)_{L^2(\R^2;(1+|x|)^{ 2\alpha}\,dx )} \\[4pt]
& \quad + N( (P_0+P_\alpha)^\perp \h_\pp (P_0+P_\alpha)^\perp -  4(P_0+P_\alpha)^\perp V (P_0+P_\alpha)^\perp)_{L^2(\R^2;(1+|x|)^{ 2\alpha}\,dx )}\, . \nonumber 
\end{align}

As in the previous section, we first prove an upper bound on the counting functions of $\h_\pp - V$ restricted to the range of $P_0^\perp$ resp.\ $(P_0+P_\alpha)^\perp$.

\begin{proposition} \label{prop-Hp-perp-1}
Let $B$ satisfy \eqref{ass-B} and let $\alpha >0$.
Assume that $V \in L^1(\R_+, L^p(\Sph))$  for some $p>1$. Then 
there exists a constant $C=C(B,p)$ such that 
\begin{equation}  \label{hp-perp-upperb}
N( (P_0+P_\alpha)^\perp\, \h_\pp\, (P_0+P_\alpha)^\perp -  (P_0+P_\alpha)^\perp\, V\, (P_0+P_\alpha)^\perp)_{L^2(\R^2;(1+|x|)^{ 2\alpha}\,dx )} \, \leq\, C\, \| V \|_{1,p} \, .
\end{equation}
Moreover, if $\alpha \not\in\Z$, then 
\begin{equation}  \label{hp-perp-upperb-bis}
N(P _0^\perp\, \h_\pp\, P_0^\perp -  P_0^\perp\, V\, P_0^\perp)_{L^2(\R^2;(1+|x|)^{ 2\alpha}\,dx )} \, \leq\, C\, \| V \|_{1,p} \, .
\end{equation}
\end{proposition}

\begin{proof}
Let $\tilde{P}=P_0$ if $\alpha \notin \mathbb{Z}$ and $\tilde{P}=P_0+ P_\alpha$ if $\alpha \in \mathbb{Z}$. As in Proposition \ref{prop-Hm-perp-1}, we first prove that 
\begin{equation} \label{prop4-eq-enough}
\langle v,  \tilde{P}^\perp\, \h_+\, \tilde{P}^\perp  v \rangle  \,  \gtrsim_{\, \alpha} \,  \int_{\R^2} (1+|x|)^{2\alpha}\big( |\nabla v|^2+   |x|^{-2} |v|^2\big)\,dx \qquad v\in D_\m(\alpha).
\end{equation}
and by density it suffices to show the above estimate for all $v\in C_0^\infty(\R^2)$. Let $\varphi= \tilde{P}^\perp\, v$. First of all, we have 
\begin{align} \label{prop4-hardy-orth}
\int_{\R^2}  (1+|x|)^{2\alpha} |\nabla \varphi |^2\,dx
&  \geq \, \int_{\R^2} (1+|x|)^{2\alpha}   |x|^{-2} |\varphi|^2\,dx \,.
\end{align}
Let $\varphi_m$ again denote the Fourier coefficients of $\varphi$, see \eqref{fi-m}. Then 
\begin{equation} \label{prop4-h-fourier}
\begin{aligned}  
 \langle \varphi,   \h_+ \varphi \rangle = \Q_+[\varphi] & = \sum_{m\in\Z} \int_0^\infty (1+r)^{2\alpha}  \big | \varphi'_m(r) + \frac{m\, \varphi_m(r)}{r} \big|^2\, r\, dr  \\[5pt]
&=   \sum_{m\in\Z}   \int_0^\infty (1+r)^{2\alpha}\,    r^{1-2m}\, \big | \partial_r( r^{m}\, \varphi _m(r))\big|^2\, dr \, .
\end{aligned}
\end{equation} 
We estimate the integrals in the sum from below. 

First, suppose $m<0$. Then, after integration by parts, we have
\begin{align} \label{prop4-m<0}
\int_0^\infty (1+r)^{2\alpha}\,   \big | \varphi_m' + \frac{m\, \varphi_m}{r} \big|^2\, r\, dr & = \int_0^\infty (1+r)^{2\alpha}\,   \Big( |\varphi'_m|^2 + \frac{m^2\, |\varphi_m|^2}{r^2} -2\alpha m \frac{ |\varphi_m|^2}{(1+r)r} \Big)\, r\, dr
\nonumber  \\[5pt]
& \geq \int_0^\infty (1+r)^{2\alpha}\,   \Big( |\varphi'_m|^2 + \frac{m^2\, |\varphi_m|^2}{r^2} \Big)\, r\, dr.
\end{align}

Now consider $m> \alpha$. Let $g_m= r^{m} \varphi _m$. Then $\liminf_{t\to 0} |g_m(t)| =0$ and  Theorem \ref{thm-class-2} applied with $U(t) = t^{1-2m } (1+t)^{2\alpha}$ 
and $W(t) = t^{-2}\, U(t)$ gives 
\begin{align*}
\int_0^\infty (1+r)^{2\alpha}\,  r^{1-2m}\, |g'_m(r)|^2\, dr\, &\gtrsim_{\, \alpha} \,  (m-\alpha)^2 \int_0^\infty (1+r)^{2\alpha}\,  r^{1-2m}\, \frac{|g_m(r)|^2}{r^2}\, dr \\
& \gtrsim_{\, \alpha} \, m^2\int_0^\infty (1+r)^{2\alpha}\,   \frac{|\varphi_m(r)|^2}{r^2} \, r\, dr .
\end{align*}

If $0<m< \alpha$, we proceed as in \cite{FaKo} and get
\begin{align*}
\int_0^\infty (1+r)^{2\alpha}\,    r^{1-2m}\, \big |g_m'(r)\big|^2\, dr \geq \int_0^1    r^{1-2m}\, \big |g_m'(r)\big|^2\, dr +  \int_1^\infty  r^{1-2m+2\alpha}\, \big |g_m'(r)\big|^2\, dr.
\end{align*}
We have by Corollary \ref{cor-class-2},
\begin{align*}
\int_0^1 r^{1-2m}\, \big |g_m'(r)\big|^2\, dr  \gtrsim_{\, \alpha} \, m^2 \int_0^1 r^{1-2m}\,  \frac{\big|g_m(r)\big|^2}{r^2} \, dr  
\end{align*}
and by Corollary \ref{cor-class-1},
\begin{align*}
 \int_1^\infty  r^{1-2m+2\alpha}\, \big |g_m'(r)\big|^2\, dr \gtrsim_{\, \alpha} \, (m-\alpha)^2  \int_1^\infty  r^{1-2m+2\alpha}\, \big |g_m'(r)\big|^2\, dr \gtrsim_{\, \alpha} \, m^2  \int_1^\infty  r^{1-2m+2\alpha}\, \frac{\big|g_m(r)\big|^2}{r^2}\, dr.
\end{align*}
As $1 \geq 2^{-2\alpha} (1+r)^{2\alpha} $ for $0 \leq r\leq 1$ and $r^{2\alpha} \geq 2^{-2\alpha} (1+r)^{2\alpha} $ for $r \geq 1$, we conclude
\begin{align*}
\int_0^\infty  (1+r)^{2\alpha}  \big | \varphi'_m(r) + \frac{m\, \varphi_m(r)}{r} \big|^2\, r\, dr &= \int_0^\infty (1+r)^{2\alpha}\,    r^{1-2m}\, \big |g_m'(r)\big|^2\, dr \\
& \gtrsim_{\, \alpha} \, m^2  \int_0^\infty (1+r)^{2\alpha} r^{1-2m}\, \frac{\big|g_m(r)\big|^2}{r^2}\, dr \\
& = m^2 \int_0^\infty (1+r)^{2\alpha}\,   \frac{|\varphi_m(r)|^2}{r^2} \, r\, dr.
\end{align*}

Hence using \eqref{young}, we see that for $m> \alpha$ and $0<m< \alpha$, we have
\begin{equation} \label{prop4-m>n}
\int_0^\infty  (1+r)^{2\alpha}  \big | \varphi'_m + \frac{m\, \varphi_m}{r} \big|^2\, r\, dr \,  \gtrsim_{\, \alpha} \, 
\int_0^\infty (1+r)^{2\alpha}\, \Big(  |\varphi'_m |^2+ \frac{m^2 |\varphi_m|^2}{r^2}\Big) \, r\, dr\,  .
\end{equation}

Now, by definition of $\tilde{P}^\perp$,
\begin{equation}  \label{prop4-fi-m-zeros}
\varphi_0 = 0 \qquad \text{and} \qquad \varphi_\alpha = 0 \quad (\text{if } \alpha \in\mathbb{Z}),
\end{equation} 
therefore \eqref{prop4-m<0} and \eqref{prop4-m>n} combined with the Parseval identity imply
\begin{equation}  \label{prop4-Q-lowerb1}
 \Q_+[\varphi]  \, \gtrsim_{\, \alpha} \,  \sum_{m\in\Z} \int_0^\infty (1+r)^{2\alpha}   \Big( |\varphi'_m|^2 + \frac{m^2\, |\varphi_m|^2}{r^2} \Big)\, r\, dr  =  \int_{\R^2}  (1+|x|)^{2\alpha}\,  |\nabla \varphi |^2\,dx\, .
 \end{equation}
Together with \eqref{prop4-hardy-orth}, this proves \eqref{prop4-eq-enough}. 

Let now $\U: L^2(\R^2,dx) \to L^2(\R^2,(1+|x|)^{2\alpha}dx)$ given by $\U \psi=  (1+|x|)^{-\alpha} \psi$. The operator $\U$ is unitary and commutes with $\tilde{P}^\perp$. Hence by \eqref{prop4-Q-lowerb1} and \eqref{prop4-hardy-orth},
\begin{equation} \label{prop4-no-weight}
\langle \psi , \U^* \tilde{P}^\bot \mathcal H_\m \tilde{P}^\bot \U\,  \psi  \rangle  \, \gtrsim_{\, \alpha} \,  \int_{\R^2} (1+|x|)^{2\alpha} |\nabla ((1+|x|)^{-\alpha }\psi)|^2\,dx \, \geq \, \int_{\R^2} \frac{|\psi|^2}{|x|^2} \,dx \, .
\end{equation} 
On the other hand, integration by parts shows that
\begin{align*}
		\int_{\R^2} (1+|x|)^{2\alpha} |\nabla ((1+|x|)^\alpha \psi)|^2\,dx
		& = \int_{\R^2} (|\nabla \psi|^2 + \alpha (1+|x|)^{-2} (|x|^{-1}+\alpha ) |\psi|^2 )\,dx \, \geq \int_{\R^2} |\nabla \psi|^2  \,dx  \, .
\end{align*}
Thus, 
$$
\langle \psi , \U^* \tilde{P}^\bot \mathcal H_\m \tilde{P}^\bot \U\,  \psi  \rangle  \, \gtrsim_{\, \alpha} \,  \int_{\R^2} \Big (|\nabla \psi|^2 + \frac{|\psi|^2}{|x|^2} \Big)\,dx \, 
$$
and there exists a constant $c_\alpha>0$ such that 
\begin{align*}
 N( \tilde{P}^\perp \h_\m \tilde{P}^\perp - \tilde{P}^\perp V \tilde{P}^\perp)_{L^2(\R^2;(1+|x|)^{\pm 2\alpha}\,dx )} \, & \leq\,  N( \tilde{P}^\perp ( -\Delta + |x|^{-2} -c_\alpha V)\tilde{P}^\perp)_{L^2(\R^2)} \\[5pt] 
 & \leq N( -\Delta + |x|^{-2} -c_\alpha V)_{L^2(\R^2)} \, .
\end{align*}
The statement of the Proposition now follows as previously from Theorem \ref{thm-ln}. 
\end{proof}

We continue with estimates on $N(P _0\, \h_\pp P_0 -  P_0VP_0)$ and $N(P _\alpha\, \h_\pp P_\alpha -  P_\alpha VP_\alpha)$. Recall that the operators $P_m$ are defined in \eqref{pm-def}.  By \eqref{new-qform}, 
\begin{equation} \label{h-fourier-bis}
\begin{aligned}  
 \langle  \varphi,  \h_\pp\,   \varphi \rangle = \Q_\pp[\varphi] & = \sum_{m\in\Z} \int_0^\infty (1+r)^{2\alpha}  \big | \varphi'_m(r) + \frac{m\, \varphi_m(r)}{r} \big|^2\, r\, dr  
 \\[5pt]
&=   \sum_{m\in\Z}   \int_0^\infty (1+r)^{2\alpha}\,    r^{1-2m}\, \big | \partial_r( r^{m}\, \varphi _m(r))\big|^2\, dr \, ,
\end{aligned}
\end{equation} 
with $\varphi_m$ given by \eqref{fi-m}. Hence denoting by $h_m^\pp$ the operator associated in $L^2(\R_+; (1+r)^{2\alpha}\, r dr)$ with the closure of the quadratic from
\begin{equation} \label{qm+form}
  \int_0^\infty   \Big| u'(r) + \frac{m u}{r} \Big|^2\, (1+r)^{2\alpha}\, r\, dr  
\end{equation} 
originally defined on $C^1_c(\R_\pp)$,  it follows that 
\begin{equation} \label{upperb-PHP-0}
N(P _0\, \h_\pp P_0 -  P_0VP_0)_{L^2(\R^2;(1+|x|)^{ 2\alpha}\,dx )}   =   N( h_0^\pp-  \V)_{L^2(\R_+; (1+r)^{2\alpha}\, r dr)}, 
\end{equation} 
with $\V$ as in \eqref{V-m-eq}. Similarly, for $\alpha\in\Z$, 
\begin{equation} \label{upperb-PHP-alpha}
N(P _\alpha\, \h_\pp P_\alpha-  P_\alpha VP_\alpha)_{L^2(\R^2;(1+|x|)^{ 2\alpha}\,dx )} =  N( h_\alpha^\pp-  \V)_{L^2(\R_+; (1+r)^{2\alpha}\, r dr)} .
\end{equation}

Let us now focus on the operators $h_m^\pp$, $m\in\mathbb{Z}$. We are particularly interested in the cases $m=0$ and $m=\alpha$ (if $\alpha\in\mathbb{Z}$). We denote 
\begin{equation}
w(r) = \begin{cases}
			 r  &\quad  \text{if}\quad  0 < r\leq 1 \,,
				\\[3pt]	
			r^{1+2\alpha}  &\quad \text{if}\quad   1 < r\, .
			\end{cases}
\end{equation}
Since
\begin{equation} \label{w-2-sided}
2^{-2\alpha} \, r (1+r)^{2\alpha}  \, \leq w(r)  \leq \,  r (1+r)^{2\alpha},
\end{equation}
we deduce from \eqref{qm+form} that 
$$
N( h_m^\pp-  \V)_{L^2(\R_+; (1+r)^{2\alpha}\, r dr)} 	\, \leq \, N( \widetilde T_m-  2^{2\alpha}\, \V)_{L^2(\R_+; w(r) dr)},
$$
where $\widetilde T_m$ is the operator in $L^2(\R_+; w(r) dr)$ generated by the quadratic form 
$ \int_0^\infty   | u'(r)|^2 \, w(r) dr  $.  The unitary mapping  
\begin{equation} \label{H+sec-unit-trans}
\U: L^2(\R_+, w(r) dr) \to L^2(\R_+)
\end{equation}
given by  $\U u=:v= \sqrt{w}\, u$, then shows that 
\begin{equation} \label{H+sec-compare-1}
N( h_m^\pp-  \V)_{L^2(\R_+; (1+r)^{2\alpha}\, r dr)} 	 \, \leq \, N( T_m-  2^{2\alpha}\, \V)_{L^2(\R_+)}
\end{equation}
where $T_m=\U\, \widetilde T_m\, \U^{-1}$. An integration by parts shows that $T_m$ is generated by the closure of the quadratic form 
\begin{equation} \label{H+sec-T-form}
	\int_0^\infty |v'(r)|^2  \,dr +\alpha |v(1)|^2 +\Big(m^2-\frac 14\Big) \int_0^1 \frac{ v(r)^2}{r^2}\, dr +  \Big((\alpha-m)^2-\frac 14\Big) \int_1^\infty \frac{ v(r)^2}{r^2}\, dr
\end{equation}
defined for $v\in C^1_c(\R_\pp)$.  Notice that the term $\alpha  |v(1)|^2$ comes from the non-derivability of $w$ at $r=1$.

Next we calculate the integral kernel of  $(T_m+\kappa^2)^{-1}$ and its limit as $\kappa \to 0$ for the cases $m=0$ and $m=\alpha$ (if $\alpha \in\mathbb{Z}$).

 By Sturm--Liouville theory, the integral kernel of $(T_m+\kappa^2)^{-1}$ can be written in terms of two solutions $v_1$ and $v_2$ satisfying
	\begin{align}
		-v'' +(m^2-\tfrac14 ) \, r^{-2} v & = -\kappa^2 v \qquad \text{in}\ (0,1) \,, \nonumber\\
		-v'' + ((\alpha-m)^2-\tfrac14) \, r^{-2} v & = -\kappa^2 v \qquad \text{in}\ (1,\infty) \,, \nonumber\\
		v(1_-) & = v(1_+)   \label{H+sec-continuity}\\ 
		v'(1_-) & = v'(1_+) - \alpha v(1_+) \label{H+sec-jump}\,.
	\end{align}
The jump condition for the derivative at $r=1$ comes from the term $\alpha |v(1)|^2$ in \eqref{H+sec-T-form}. The function $v_1$ is supposed to lie in the form domain of $T_m$ near the origin and $v_2$ is supposed to be square-integrable at infinity.

Using standard facts about Bessel's equation \cite[Sec.~9]{as}, we find that these two solutions are given by 
	\begin{align} \label{H+sec-v-1}
		v_1(r) & = \sqrt{r} \times \begin{cases}
			I_m(\kappa r)  & \text{if }\quad   0 < r\leq 1 \,,
			\\[5pt] 
		A_m(\kappa) I_{\alpha-m}(\kappa r) +  B_{m}(\kappa) K_{\alpha-m}(\kappa r) & \text{if }\quad  1 < r < \infty \,,
		\end{cases}
	\end{align}
	and 
	\begin{align}  \label{H+sec-v-2} 
		v_2 (r) & = \sqrt{r} \times \begin{cases}
			 D_m(\kappa) I_{m}(\kappa r) +  \ C_m(\kappa) K_m(\kappa r)  & \text{if }\quad   0 < r\leq 1 \,,
			\\[5pt] 
			K_{\alpha-m}(\kappa r) & \text{if }\quad  1 < r < \infty \,,
		\end{cases}
	\end{align}
	with coefficients $A_m(\kappa),  B_m(\kappa),C_m(\kappa)$ and $ D_m(\kappa)$ that are determined by matching conditions \eqref{H+sec-continuity} and \eqref{H+sec-jump}.  From the Wronski relation \cite[Eq.~9.6.15]{as} for the Bessel functions, 
	\begin{equation}\label{H+sec-wronski}
		W\big\{ K_\nu(z), I_\nu(z)\big\} = I _\nu(z) K_{\nu+1} (z) +K_\nu(z) I_{\nu+1} (z) = \frac 1z \,,
	\end{equation}
it follows that
\begin{equation} \label{H+sec-wronski-2}
	W\big\{ v_2, v_1 \big\} = C_m(\kappa) =  A_m(\kappa)  = - W\big\{ v_1, v_2 \big\}\, .
\end{equation}
Inserting \eqref{H+sec-v-1}, \eqref{H+sec-v-2} into \eqref{H+sec-continuity}, \eqref{H+sec-jump}, and using \eqref{H+sec-wronski} we obtain
\begin{equation}\label{H+sec-coef}
	\begin{aligned} 
		A_m(\kappa) & = \kappa I_m'(\kappa) K_{\alpha-m}(\kappa) +\alpha I_m(\kappa) K_{\alpha-m}(\kappa) -\kappa I_m(\kappa) K_{\alpha-m}'(\kappa)  \,, \\[2pt]
		B_m(\kappa) & = -\kappa I_m'(\kappa) I_{\alpha-m}(\kappa) - \alpha I_m(\kappa) I_{\alpha-m}(\kappa)+ \kappa I_m(\kappa) I_{\alpha-m}'(\kappa)  \,, \\[2pt]
		D_m(\kappa) & = -\kappa K_m'(\kappa) K_{\alpha-m}(\kappa) - \alpha K_m(\kappa) K_{\alpha-m}(\kappa)+ \kappa K_m(\kappa) K_{\alpha-m}'(\kappa)  \, .
	\end{aligned}
\end{equation}
Moreover, since  $T_m$ is a nonnegative operator, it has no negative eigenvalues. Consequently we must have $A_m(\kappa)\neq 0$ for all $\kappa>0$.
Using the Sturm--Liouville theory, we deduce that for any $\kappa>0$ and $r \leq r'$, 
	\begin{align}\label{H+sec-green}
		\begin{aligned} 
			(T_m+\kappa^2)^{-1}(r,r')&= -\frac{v_1(r)v_2(r')}{W\{v_1,v_2\}} \\[5pt]
			&=\sqrt{r r'} \times
			\begin{cases}
				 I_m(\kappa r) K_m(\kappa r')+  f_m(\kappa)  I_m(\kappa r) I_m(\kappa r')  & \text{if}\  0 < r\leq r'\leq  1 \,,
				\\[5pt] 
				 I_{\alpha-m}(\kappa r) K_{\alpha-m}(\kappa r')+  g_m(\kappa)  K_{\alpha-m}(\kappa r) K_{\alpha-m}(\kappa r')  & \text{if}\    1 < r \leq r' \,, \\[5pt]
				A^{-1}_m(\kappa)\, I_m(\kappa r)  K_{\alpha-m}(\kappa r')  & \text{if}\    0 < r \leq 1 \leq r' \,,
			\end{cases}
		\end{aligned}
	\end{align}
	where we have denoted 
	$$
 f_m(\kappa) := \frac{D_m(\kappa) }{ A_m(\kappa)} \qquad \text{and}\qquad 
	g_m(\kappa) := \frac{ B_m(\kappa) }{ A_m(\kappa)} \, .
	$$

For $m=0$, the limit $\kappa \to 0$ yields

\begin{lemma} \label{lem-T-0}
Let $\alpha>0$. Then for any $r,r'\in\R_+$,
\begin{equation} \label{limit-kernel-0}
T_0^{-1}(r,r') := \lim_{\kappa\to 0} (T_0+\kappa^2)^{-1}(r,r') = 
 \sqrt{r r'} \times
			\begin{cases}
			 \frac{1}{2\alpha} -\log ( r')		 & \text{if}\quad  0 < r\leq r'\leq  1 \,,
				\\[5pt] 
				 \frac{1}{2\alpha}\  ( r/r')^\alpha    & \text{if}\quad    1 < r \leq r' \,, \\[5pt]
				 \frac{1}{2\alpha}\  (r')^{-\alpha}& \text{if}\quad   0 < r \leq 1 \leq r' \, .
			\end{cases}
\end{equation}
As usual, the formula for $r> r'$ follows by interchanging the variables. 
\end{lemma}

\begin{proof}
With \eqref{H+sec-coef} and \eqref{KI-zero}
we obtain 
\begin{align*}
f_0(\kappa) &= \frac{1}{2\alpha} - K_0(\kappa) + o(1)  \qquad\text{as}\ \kappa\to 0, \\
g_0(\kappa) &= o(\kappa^{2\alpha}) \qquad\qquad\qquad\quad \ \text{as}\ \kappa\to 0.
\end{align*}
Hence, using \eqref{H+sec-green} and the asymptotic expansions \eqref{KI-zero} again, if $0<r\leq r'\leq 1$, then
$$
T_0^{-1}(r,r') = \lim_{\kappa\to 0} \big( K_0(\kappa r') I_0(\kappa r) +  f_0(\kappa)  I_0(\kappa r) I_0(\kappa r')  \big) = \frac{1}{2\alpha} -\log ( r')	.
$$
The other two identities in \eqref{limit-kernel-0} are obtained in a similar way. 
\end{proof}

The case $m=\alpha$ yields a similar result.

\begin{lemma} \label{lem-T-alpha}
Let $0<\alpha\in\Z$. Then for any $r,r'\in\R_+$,
\begin{equation} \label{limit-kernel-alpha}
T_\alpha^{-1}(r,r') := \lim_{\kappa\to 0} (T_\alpha+\kappa^2)^{-1}(r,r') = 
 \sqrt{r r'} \times
			\begin{cases}
			        \frac{1}{2\alpha}\  ( r/r')^\alpha   & \text{if}\quad  0 < r\leq r'\leq  1 \,,
				\\[5pt] 
			 	\frac{1}{2\alpha}  + \log ( r)	 & \text{if}\quad    1 < r \leq r' \,, \\[5pt]
		    		\frac{1}{2\alpha} \ (r)^{\alpha}	& \text{if}\quad   0 < r \leq 1 \leq r' \, .
			\end{cases}
\end{equation}
The formula for $r> r'$ is obtained by interchanging the variables. 
\end{lemma}

\begin{proof}
As in the proof of Lemma \ref{lem-T-0} we use asymptotic expansions \eqref{KI-zero} to find 
\begin{align*}
f_\alpha(\kappa) &= o(\kappa^{-2\alpha})\qquad\text{as}\ \kappa\to 0, \\
g_\alpha(\kappa) &= -\frac{1}{K_0(\kappa)}+\frac{1}{2\alpha K_0^2(\kappa)} + o\big(K_0^{-2}(\kappa)\big)\qquad\text{as}\ \kappa\to 0,
\end{align*}
and, consequently, for $1<r\leq r'$,
$$
T_\alpha^{-1}(r,r') = \lim_{\kappa\to 0} \big(K_0(\kappa r') I_0(\kappa r) +  g_\alpha(\kappa)  K_0(\kappa r) K_0(\kappa r')  \big) = \frac{1}{2\alpha}+\log (r) \, .
$$
 The remaining parts of equation \eqref{limit-kernel-alpha} follow in the same way.
\end{proof}

From the previous two lemmas, we deduce

\begin{proposition} \label{prop-H+radial+alpha}
Let $B$ satisfy \eqref{ass-B} and suppose that $\alpha>0$.
\begin{enumerate}
\item If $V \in L^1_{\rm loc}(\R^2)$ and $V  \log |x  | \in L^1(\B_1)$, then 
there exists a constant $C=C(B)$ such that 
\begin{equation} 
N(P _0\, \h_\pp P_0-  P_0 V P_0) \, \leq\,  C \int_{\R^2}  V(x) \big(1+ \id_{\!\B_1}(x) | \log |x || \big)\, dx .
\end{equation}
\item If $\alpha\in \Z$ and $V \in L_{\rm loc}^1(\R^2)$ and $V  \log |x  | \in L^1(\R^2)$, then 
there exists a constant $C=C(B)$ such that 
\begin{equation} 
N(P _\alpha\, \h_\pp P_\alpha-  P_\alpha V P_\alpha) \, \leq\,  C \int_{\R^2}  V(x) \big(1+\id_{\R^2\setminus\B_1} (x)| \log |x || \big)\, dx .
\end{equation}
\end{enumerate}

\end{proposition}

\begin{proof} (1) The integral kernel  $\sqrt{\V(r)} \ T_0^{-1}(r,r')\, \sqrt{\V(r')}$ is positive-definite, see Lemma \ref{lem-pos-def-0}. Hence, arguing as in the proof of Lemma \ref{lem-01}, we conclude with the Birman-Schwinger principle
\begin{equation} \label{t0-upperb}
N( T_0 -  \V)_{L^2(\R_+)}\,    \leq\,   \tr \big(\sqrt{\V} \ T_0^{-1}\, \sqrt{\V} \, \big) \, \leq C \int_0^\infty  \V(r)\, (1+\id_{(0,1)}  (r) | \log r|\, ) r\, dr\, .
\end{equation}
The claim thus follows from \eqref{upperb-PHP-0}, \eqref{H+sec-compare-1} and \eqref{V-m-eq}. 

(2) In view of Lemma \ref{lem-pos-def-alphs} the integral kernel  $\sqrt{\V(r)} \ T_\alpha^{-1}(r,r')\, \sqrt{\V(r')}$ is positive-definite, hence 
\begin{equation} \label{t-alpha-upperb}
N( T_\alpha -  \V)_{L^2(\R_+)}\,    \leq\,   \tr \big(\sqrt{\V} \ T_\alpha^{-1}\, \sqrt{\V} \, \big) \leq C\,  \int_0^\infty  \V(r)\, (1+\id_{(1,\infty)}  (r) | \log r|\, ) r\, dr\, .
\end{equation}
The claim then follows from \eqref{upperb-PHP-alpha}, \eqref{H+sec-compare-1} and \eqref{V-m-eq}.
\end{proof}

Combining Propositions \ref{prop-Hp-perp-1} and \ref{prop-H+radial+alpha}, we find

\begin{proposition}\label{prop-H+}
Let $B$ satisfy \eqref{ass-B} and suppose that $\alpha>0$.
\begin{enumerate} 
\item Assume that $\alpha \not\in\Z$. Then for any $p>1$ there exist constants $C_1=C_1(B,p)$ and $C_2=C_2(B)$ such that 
\begin{equation} 
N(  H_\pp\, -   V) \, \leq\,  C_1\, \|V\|_{1,p} +C_2\,  \| V  \log |x|\|_{L^1(\!\B_1)} 
\end{equation}
for all  $V \in L^1(\R_+, L^p(\Sph))$ with  $V  \log |\cdot | \in L^1(\B_1)$.

\item Assume that $\alpha \in\Z$. Then for any $p>1$ there exist constants $\CC_1=\CC_1(B,p)$ and $\CC_2=\CC_2(B)$ such that 
\begin{equation} 
N(  H_\pp\, -   V) \, \leq\,   \CC_1\, \|V\|_{1,p}  +\CC_2\,  \| V\log |x|\|_{L^1(\R^2)} 
\end{equation}
for all  $V \in L^1(\R_+, L^p(\Sph))$ with  $V  \log |\cdot | \in L^1(\R^2)$.
\end{enumerate}
\end{proposition}

\begin{proof}
 This follows from equations \eqref{passage-weighted},  \eqref{split-2a} and \eqref{split-2b} combined with Propositions \ref{prop-Hp-perp-1} and \ref{prop-H+radial+alpha}.
\end{proof}


\section{\bf Strong coupling asymptotic}
\label{sec-strong-coupling}

In this section we are going to discuss the behavior of $N( H_\ppm -\lambda V)$ and $N\big ((i\nabla +A)^2-  \lambda V\big)$ in the limit $\lambda\to \infty$. 
We start by the case of regular fast decaying potentials $V$ in which the strong coupling asymptotic of the counting function displays the semi-classical 
behavior.

\subsection*{Semi-classical behavior} If $B$ and $V$ are bounded and compactly supported, then
\begin{equation} \label{strong-coupling}
N\big ((i\nabla +A)^2-  \lambda V\big) ,\; N( H_\m -\lambda V) , \; N( H_\pp -\lambda V) =
\frac{\lambda}{2\pi} \, \int_{\R^2} V(x)_\pp\, dx  + o(\lambda)\, 
\end{equation}
as $\lambda\to \infty$. This follows from \cite[Theorem.~1.1 \and Remark 1.2]{rai}, see also \cite{FK}.

\subsection*{Non semi-classical behavior} As mentioned in Section \ref{sec-intro}, the strong coupling asymptotic of $N (\PP - \lambda V)$
might display a non-semiclassical behavior even for potentials in $L^1(\R^2)$. 

\begin{proposition}[Slowly decaying potentials] \label{prop-semiclass-1}
Let $B$ satisfy assumption \eqref{ass-B} and assume $0<\alpha \in \mathbb{Z}$. Let $V\in L^q_{\rm loc}(\R^2),  q>1$ and assume that 
\begin{equation}  \label{cond-asymp}
V(x) - W_\sigma(x) = o(W_\sigma(x)) \qquad |x|\to \infty, 
\end{equation} 
for some $\sigma >1$. Then
\begin{equation} 
N\big (  H_\pm -\lambda V\big) \, \asymp\, \lambda^\sigma \qquad \text{as } \lambda\to \infty.
\end{equation}
\end{proposition}

\begin{proof}
In view of \eqref{passage-weighted} it suffices to prove the claim for $N(  \h_\pm -\lambda V)$. Let us consider first the operator $\h_\pp$.  Applying \eqref{schwarz-eps} with 
$\Pi=P_\alpha$ and $\eps<1$ we find
\begin{equation} \label{two-sided}
N\big (  P_\alpha (\h_\pp -(1-\eps) \lambda V_\pp) P_\alpha \big) \leq N\big (  \h_\pp -\lambda V\big)  \leq  N\big (  P_\alpha (\h_\pp -(1+\eps) \lambda V_\pp) P_\alpha \big) 
+ N\big (  P^\perp_\alpha (\h_\pp -(1+\eps^{-1}) \lambda V) P^\perp_\alpha \big) .
\end{equation} 
Note that $W_\sigma\in  L^1(\R_+, L^p(\Sph))$ for any $\sigma >1$ and for any $p>1$. Hence by assumptions on $V$  we have $V_\pp \in L^1(\R_+, L^p(\Sph))$ for any $p >1$ and $V_\pp \log |\cdot | \in L^1_{\rm loc}(\R^2)$. 
Equation \eqref{hp-perp-upperb} in combination with Proposition
\ref{prop-H+radial+alpha} (1) then gives
$$
 \limsup_{\lambda\to \infty} \lambda^{-\sigma}  N\big (  P^\perp_\alpha (\h_\pp -(1+\eps^{-1}) \lambda V_\pp) P^\perp_\alpha \big) =0
$$
for any $0<\eps<1$.
As for the first term in \eqref{two-sided}, we note that 
\begin{equation}
N\big (  P_\alpha (\h_\pp - \lambda V_\pp ) P_\alpha \big) =  N( h_\alpha^\pp- \lambda \V_\pp )_{L^2(\R_+; (1+r)^{2\alpha}\, r dr)}, 
\end{equation}
where $h_\alpha^\pp$ is the operator in $L^2(\R_+; (1+r)^{2\alpha}\, r dr)$ associated with the quadratic form \eqref{qm+form} for $m=\alpha$, and where 
$$
\V_\pp = \frac{1}{2\pi} \int_0^{2\pi}\, V_\pp(r,\theta)\, d\theta\, .
$$
Repeating the reasoning of the proof of Proposition \ref{prop-Hm-perp-1}, see in particular equations \eqref{DBC-1} and \eqref{reduced-weights}, we find that 
\begin{align} \label{DBC-2}
0\leq N( h_\alpha^+ -  \V_\pp)_{L^2(\R_+; (1+r)^{2\alpha}\, r dr)}- N(  \mathfrak{h}_\alpha^\pp -  \V_\pp)_{L^2((0,1); (1+r)^{2\alpha}\, r dr)}    -  N( \mathfrak{h}_\alpha^\pp -  \V_\pp)_{L^2((1,\infty); (1+r)^{2\alpha}\, r dr)}   &\leq 1,
\end{align} 
where the  operator $ \mathfrak{h}_\alpha^\pp$ acts in $L^2((0,1), (1+r)^{2\alpha}\, r dr)$ respectively $L^2((1,\infty), (1+r)^{2\alpha}\, r dr)$  as $h_\alpha^\pp$ with additional Dirichlet boundary condition at $r=1$. 
Similarly as in \eqref{reduced-weights} we replace the integral weight $(1+r)^{2\alpha}$ by $1$ on $(0,1)$ and by $r^{2\alpha}$ on $(1,\infty)$. This implies 
\begin{equation} \label{reduced-weights-2}
\begin{aligned} 
N( h_{\alpha,1}^\pp -  4^{-\alpha}\, \V_\pp)_{L^2((0,1); \, r dr)} \, \leq\, N( \mathfrak{h}_\alpha^\pp -  \V_\pp)_{L^2((0,1); (1+r)^{2\alpha}\, r dr)}  \,& \leq \, N( h_{\alpha,1}^\pp -  4^\alpha\, \V_\pp)_{L^2((0,1); \, r dr)}  \\[6pt]
N( h_{\alpha,2}^\pp -  4^{-\alpha}\, \V_\pp)_{L^2((1,\infty); \, r^{1+2\alpha} dr)} \leq N( \mathfrak{h}_\alpha^\pp -  \V_\pp)_{L^2((1,\infty); (1+r)^{2\alpha}\, r dr)}  \, &\leq \, N(h_{\alpha,2}^\pp -  4^\alpha\, \V_\pp)_{L^2((1,\infty); \, r^{1+2\alpha} dr)} ,
\end{aligned}
\end{equation} 
where the operators $h_{\alpha,1}^\pp$ and $h_{\alpha,2}^\pp$ 
are associated with quadratic forms 
\begin{align} 
q_{\alpha,1}^\pp[v]  &=  \int_0^1   \big | \partial_r( r^{\alpha}\, v(r))\big|^2\, r^{1-2\alpha}\, dr , \quad \ v\in H^1((0,1), r dr), \qquad  \ \quad v(1)=0,
\\[4pt]
q_{\alpha,2}^\pp[v] &=  \int_1^\infty   \big | \partial_r( r^\alpha\, v(r))\big|^2\, r\, dr ,  \quad v\in H^1((1,\infty), r^{1+2\alpha} dr),  \qquad  \ \ v(1)=0.
\end{align} 

Note that the operator $h_{\alpha,1}^\pp$ is identical with the operator $h_{-\alpha,1}^-$ defined in \eqref{q-forms-dirichlet}. By Lemma \ref{lem-01}, 
$$
 \limsup_{\lambda\to \infty} \lambda^{-\sigma} N( h_{\alpha,1}^\pp -  \lambda \V_\pp)_{L^2((0,1); \, r dr)} = 0. 
$$
On the other hand, the application of the transform $\U v= r^{\alpha+\frac 12}\, v$, which maps $L^2(\R_+, r^{1+2\alpha} dr)$  unitarily onto $L^2(\R_+)$ gives  
$$
 N(h_{\alpha,2}^\pp -  \V_\pp)_{L^2((1,\infty); \, r^{1+2\alpha} dr)}  =  N\big( -\partial_r^2- \frac{1}{4 r^2} -  \V_\pp\big)_{L^2(1,\infty)} 
$$
Since $W_\sigma \geq 0$, it follows that $V_\pp$ satisfies condition \eqref{cond-asymp}. Hence 
\begin{equation} 
\lim_{\lambda\to \infty} \lambda^{-\sigma} N\big( -\partial_r^2- \frac{1}{4 r^2} - \lambda \V_\pp\big)_{L^2(1,\infty)}  = \frac{\Gamma\big(\sigma -\frac 12\big)}{2\sqrt{\pi} \ \Gamma(\sigma)}\, ,
\end{equation} 
see \cite[Sec.~4.4 and Prop.~6.1(b)]{bl}. Upon inserting the above estimates into equations \eqref{DBC-2} and  \eqref{two-sided} we thus get 
\begin{equation}  \label{h-plus-semicl}
0< \liminf_{\lambda\to \infty} \lambda^{-\sigma} N\big (   \h_\pp -\lambda V\big) \leq \limsup_{\lambda\to \infty} \lambda^{-\sigma} N\big ( \h_\pp -\lambda V\big ) < \infty
\end{equation} 
as claimed.

Next we consider the operator $\h_\m$. As quadratic forms, 
$$
\h_\pp - 2|B| \, \leq \, \h_\m  \leq \, \h_\pp + 2|B|\, .
$$
Hence for any $\lambda\geq 1$,
$$
N\big (   \h_\pp -\lambda (V-2 |B|) \big)\, \leq \, N\big (   \h_\m -\lambda V\big) \, \leq\, N\big (   \h_\pp -\lambda (V+2|B|) \big)
$$
Since $V$ satisfies \eqref{cond-asymp}, Assumption \eqref{ass-B} ensures that so does $V\pm 2 |B|$. The claim now follows from \eqref{h-plus-semicl}.
\end{proof}

\begin{proposition}[Potentials with local singularities] \label{prop-semiclass-2} Let $B$ satisfy assumption \eqref{ass-B} and assume moreover that 
$B\in L^\infty(\R^2)$. 
Let $V \in L^1(\R_+, L^p(\Sph))$ for some $p>1$ and suppose that  $V  \log |\cdot | \in L^1(\R^2\setminus\B_1)$.  
\begin{equation}  \label{cond-asymp-2}
V(x) - V_\sigma(x) = o(V_\sigma(x)) \qquad |x|\to 0, 
\end{equation} 
for some $\sigma >1$, then
\begin{equation} 
N\big (  H_\pm -\lambda V\big) \, \asymp\, \lambda^\sigma \qquad \lambda\to \infty,
\end{equation}
\end{proposition}

\begin{proof} 
By assumption we have $B(x) =o(V_\sigma(x))$ as $|x|\to 0$. Hence, as above, it suffices to prove the statement for the operator $\h_\pp$. We mimic the proof of Proposition \ref{prop-semiclass-1} and apply  \eqref{schwarz-eps} with 
$\Pi=P_0$. This gives 
\begin{equation} \label{two-sided-bis}
N\big (  P_0 (\h_\pp -(1-\eps) \lambda V_\pp) P_0 \big) \leq N\big (  \h_\pp -\lambda V\big)  \leq  N\big (  P_0 (\h_\pp -(1+\eps) \lambda V_\pp) P_0 \big) 
+ N\big (  P^\perp_0 (\h_\pp -(1+\eps^{-1}) \lambda V_\pp) P^\perp_0 \big) 
\end{equation} 
for all $0 <\eps <1$. From Propositions \ref{prop-Hp-perp-1}, \ref{prop-H+radial+alpha} (2) and assumptions on $V$ we deduce that 
$$
\limsup_{\lambda\to \infty} \lambda^{-\sigma} \, N\big (  P^\perp_0 (\h_\pp - \lambda V_\pp) P^\perp_0 \big) =0.
$$
Moreover, 
\begin{equation}
N\big (  P_0(\h_\pp - \lambda V_\pp ) P_0 \big) =  N( h_0^\pp-  \lambda \V_\pp )_{L^2(\R_+; (1+r)^{2\alpha}\, r dr)}, 
\end{equation}
where $h_0^\pp$ is the operator in $L^2(\R_+; (1+r)^{2\alpha}\, r dr)$ associated with the quadratic form \eqref{qm+form} for $m=0$. As in the proof of Proposition  \ref{prop-semiclass-1} 
we thus find that
\begin{align} \label{DBC-3}
0\leq N( h_0^+ +  \V_\pp)_{L^2(\R_+; (1+r)^{2\alpha}\, r dr)}- N(  \mathfrak{h}_0^\pp -  \V_\pp)_{L^2((0,1); (1+r)^{2\alpha}\, r dr)}    -  N( \mathfrak{h}_0^\pp -  \V_\pp)_{L^2((1,\infty); (1+r)^{2\alpha}\, r dr)}   &\leq 1,
\end{align} 
where the  operator $ \mathfrak{h}_0^\pp$  acts in $L^2((0,1), (1+r)^{2\alpha}\, r dr)$ respectively $L^2((1,\infty), (1+r)^{2\alpha}\, r dr)$ as $h_0^\pp$ with additional Dirichlet boundary condition at $r=1$. 
By variational principle,
$$
N( \mathfrak{h}_0^\pp -  \V_\pp)_{L^2((1,\infty); (1+r)^{2\alpha}\, r dr)}    \leq N( h_0^\pp-  \V_\pp \id_{(1,\infty)})_{L^2(\R_+; (1+r)^{2\alpha}\, r dr)} 
$$
with $h_0^\pp$ defined by the quadratic form \eqref{qm+form}. From equations \eqref{H+sec-compare-1}, \eqref{t0-upperb} and assumptions on $V$ we deduce 
$$
\limsup_{\lambda\to \infty} \lambda^{-\sigma} \, N( \mathfrak{h}_0^\pp -  \V_\pp)_{L^2((1,\infty); (1+r)^{2\alpha}\, r dr)}    =0.
$$
It remains to consider $N(  \mathfrak{h}_0^\pp -  \V_\pp)_{L^2((0,1); (1+r)^{2\alpha}\, r dr)}$. Note that $(1+r)^{2\alpha} \asymp 1$ on $(0,1)$. Hence
\begin{equation}
N(  \mathfrak{h}_0^\pp -  \lambda  \V_\pp)_{L^2((0,1); (1+r)^{2\alpha}\, r dr)} \asymp N(  \mathfrak{h}_0^\pp -  \lambda  \V_\pp)_{L^2((0,1); r dr)} =N\big( -\partial_r^2- \frac{1}{4 r^2} - \lambda \V_\pp\big)_{L^2(0,1)} 
\end{equation}
as $\lambda \to\infty$. 
Since 
\begin{equation} 
\lim_{\lambda\to \infty} \lambda^{-\sigma} N\big( -\partial_r^2- \frac{1}{4 r^2} - \lambda \V_\pp\big)_{L^2(0,1)}  = \frac{\Gamma\big(\sigma -\frac 12\big)}{2\sqrt{\pi} \ \Gamma(\sigma)}\, ,
\end{equation} 
see \cite[Secs.~4.4 and 6.5]{bl}, the claim follows from equations \eqref{two-sided-bis}-\eqref{DBC-3}.
\end{proof}


\section{\bf Long range potentials}
\label{sec-slow}
Here we show how estimate \eqref{clr-pauli-2} can be modified in order to cover also slowly decaying such as $W_\sigma$. The proof is based on a variation of the 
method of \cite{flr}. We are indebted to Rupert Frank for suggesting us to treat this problem.

In order to state the result we need some additional notation. Let
\begin{equation} \label{weight-hardy}
w(r) = \frac{1}{1+r^2 (\log r)^2}\, , \qquad r>0.
\end{equation}
Given $a>0$, we set 
\begin{equation} \label{V-ap}
[ V]_{a} = \sup_{t>0}\, t^{1+a} \int_{\big\{\frac{\V(r)}{w(r)} >t\big\} } w(r)\, (1+|\log r|)\, r dr,
\end{equation}
with $\V$ given by \eqref{V-m-eq}.
We then have

\begin{theorem} \label{pauli-long-range}
Let $B$ satisfy \eqref{ass-B} and assume that $0 < \alpha \in\Z$. Then for any $a>0$ and any $p>1$ there exist constants $C_1(B,p)$ and  $C_2(B,a)$ such that 
\begin{equation}  \label{eq-long-range}
N(  \PP-   V) \, \leq\,  m(\alpha)+ C_1(B,p) \| V \|_{1,p}  + C_2(B,a) \, [ V]_{a} 
\end{equation}
for all  $V \in L^1(\R_+, L^p(\Sph)) \cap L^1_{\rm loc}(\R^2,  |\log |x||\, dx)$ for which the right hand side is finite.  
\end{theorem}

\begin{proof}
Recall that under the assumptions of Theorem \ref{pauli-long-range} we have $m(\alpha) =\alpha+1$. 
We note again that by \eqref{passage-weighted} it suffices to prove the claim for $N(  \h_\pm -\lambda V)$.
By Proposition \ref{prop-Hp-perp-1} and \eqref{schwarz-eps},  
\begin{equation} \label{p-0-alpha}
\begin{aligned} 
N(  \h_\pp -  V)\, &\leq\, C  \| V \|_{1,p}  +
N( (P_0+P_\alpha) \h_\pp (P_0+P_\alpha) -  2(P_0+P_\alpha)\, V\, (P_0+P_\alpha)) \\[5pt]
& \leq\,  C  \| V \|_{1,p}  + N(P _0 \h_\pp P_0 -  4 \V) + N(P _\alpha \h_\pp P_\alpha -  4 \V) \, .
\end{aligned}
\end{equation}
By \cite{FaKo} there exists a constant $C_h>0$ such that 
\begin{equation} \label{hardy}
P_0 \h_\pp P_0  \geq C_h\, w ,\qquad P_\alpha \h_\pp P_\alpha  \geq C_h\, w 
\end{equation}
in the sense of quadratic forms on $ P_0H^1(\R^2)$ and $P_\alpha H^1(\R^2)$ respectively. Now, following \cite{flr}, given $\eta \in (0,1)$  we set $s =\eta\, C_h$  and  estimate 
the operator $P_0 \h_\pp P_0 -   \V$ in the sense of quadratic forms as follows;
\begin{align} 
P_0 \h_\pp P_0-   \V &= \eta ( P_0 \h_\pp P_0-  C_h w) +(1-\eta) \big( P_0 \h_\pp P_0-  (1-\eta)^{-1} w (w^{-1} \V-s)\big) \nonumber \\[4pt]
 & \geq (1-\eta) \big( P_0 \h_\pp P_0-  (1-\eta)^{-1} w (w^{-1} \V-s)_+\big) , \label{eta-1}
\end{align}
where we have used inequality \eqref{hardy}.
Hence, 
$$
N(P _0 \h_\pp P_0 -  \V) \, \leq \, N(P _0 \h_\pp P_0 -  (1-\eta)^{-1} w (w^{-1} \V-s)_+),
$$
which by \eqref{t0-upperb} implies 
\begin{equation*} 
N(P _0 \h_\pp P_0 -  \V) \, \leq \, C (1-\eta)^{-1}  \int_0^\infty w(r) (w(r)^{-1} \V(r)-s)_+)\,  (1+\id_{(0,1)} |\log r|)\, r \, dr.
\end{equation*}
Applying the same argument to the operator $P_\alpha \h_\pp P_\alpha$ and using equation \eqref{t-alpha-upperb} we get 
\begin{equation*} 
N(P _\alpha \h_\pp P_\alpha -  \V) \, \leq \, C (1-\eta)^{-1}  \int_0^\infty w(r) (w(r)^{-1} \V(r)-s)_+)\,  (1+\id_{(1,\infty)} |\log r|)\, r \, dr.
\end{equation*}
All together, 
\begin{equation} \label{upperb-h0}
N(P _0 \h_\pp P_0 -  \V) +N(P _\alpha \h_\pp P_\alpha -  \V) \, \leq \, C (1-\eta)^{-1}  \int_0^\infty w(r) (w(r)^{-1} \V(r)-s)_+)\,  (1+ |\log r|)\, r \, dr.
\end{equation}
The layer cake representation, \cite[Sec.~1.13]{LL},  gives
\begin{align*}
\int_0^\infty\!\! w(r) (w(r)^{-1} \V(r)-s)_+)\,  (1+ |\log r|)\, r \, dr & = \int_0^\infty\!\! \int_{\{w(r)^{-1} \V(r)>s+\sigma\}} w(r) (1+ |\log r|)\, r \, dr d\sigma\\[5pt]
&\leq  \int_0^\infty  \sup_{t>0} \Big( t^{1+a} \int_{\big\{\frac{\V(r)}{w(r)} >t\big\} } w(r)\, (1+|\log r|)\, r dr \Big)\, (s+\sigma)^{-a-1}\, d\sigma\\[5pt]
&= a^{-1} s^{-a}\,  [ V]_{a}\, .
\end{align*}
Thus, in view of \eqref{p-0-alpha}, 
\begin{equation}\label{long-range-h-plus}
N(  \h_\pp -  V)\, \leq\, C  \| V \|_{1,p}  + C(B,a) \, [ V]_{a}\, .
\end{equation}

We now turn our attention to $\h_\m$. By  Proposition \ref{prop-Hm-perp-1} and equations   \eqref{split-1}, \eqref{split-n}, 
\begin{equation}  \label{1-aux}
N(  \h_\m -  V)\, \leq\, C(B,p)\,   \| V \|_{1,p}  + \sum_{m=0}^\alpha N(P_m \h_\m P_m -c_\alpha \V)_{L^2(\R_+; (1+r)^{-2\alpha}\, r dr)}\, .
\end{equation}
Combined with \eqref{DBC-1} and \eqref{reduced-weights} this further implies 
\begin{equation} \label{2-aux}
\begin{aligned}
N(  \h_\m -  V)\, &  \leq\, \alpha+1+ C(B,p)\,   \| V \|_{1,p}  \\[5pt]
& \ \ + \sum_{m=0}^\alpha \Big( N( h_{m,1}^- -  c_\alpha4^\alpha\, \V)_{L^2((0,1); \, r dr)} +N( h_{m,2}^- -  c_\alpha4^\alpha\, \V)_{L^2((1,\infty); \, r^{1-2\alpha} dr)}\Big) \, .
\end{aligned}
\end{equation}
Recall that the operators $h_{m,1}^-$ and $h_{m,2}^-$, associated to quadratic forms \eqref{q-forms-dirichlet}, act on ${L^2((0,1); \, r dr)}$ respectively ${L^2((1,\infty); \, r^{1-2\alpha} dr)}$ with additional Dirichlet boundary condition at $r=1$. Thus, by classical weighted one-dimensional Hardy inequalities, see  e.g.~\cite{muck}, the estimates \eqref{hardy} continues to hold for the operators  $h_{m,1}^-$ and $h_{m,2}^-$.  
More precisely, there exist a constant $c_h$, independent of $m$, such that 
 \begin{equation}  \label{hardy-muck} 
 h_{m,1}^- \, \geq\, c_h\, w \qquad \text{and}\qquad  h_{m,2}^- \, \geq\, c_h\, w, \qquad m=0, \dots \alpha,
 \end{equation}
in the sense of quadratic forms on $H^1((0,1), r dr),$ respectively $H^1((1,\infty), r^{1-2\alpha} dr)$ with Dirichlet boundary conditions at $r=1$. Indeed, the first inequality 
in \eqref{hardy-muck} is obvious; owing to the Dirichlet boundary condition at $r=1$ the operator $ h_{m,1}^- $ has discrete spectrum consisting of positive eigenvalues. The 
second inequality in \eqref{hardy-muck} follows from~\cite[Thm.~1]{muck} applied on the interval $(1,\infty)$ with $p=2, V(r) = r^{m+\frac 12} (r+1)^{-\alpha}$, $U(r) = V(r)\, \sqrt{ w(r)}$, and 
with 
$$
r^{-m} v(r)= \int_1^r f(t) \, dt\, ,
$$
see the second equation in  \eqref{q-forms-dirichlet}. Note also that the logarithmic term in \eqref{weight-hardy} is needed only when $m=\alpha$.

Since
 $$
 N( h_{m,1}^- - \V)_{L^2((0,1); \, r dr)} +N( h_{m,2}^- -  \V)_{L^2((1,\infty); \, r^{1-2\alpha} dr)}  \, \leq \, C_m  \int_0^\infty \V(r)\,  (1+ |\log r|)\, r \, dr,
 $$
 cf.~Lemmas \ref{lem-01}, \ref{lem-1infty} and \ref{lem-1infty-bis}, we use the same arguments as in \eqref{eta-1}-\eqref{upperb-h0}, and keeping in mind \eqref{hardy-muck} we arrive at
  $$
  N( h_{m,1}^- - \V)_{L^2((0,1); \, r dr)} +N( h_{m,2}^- -  \V)_{L^2((1,\infty); \, r^{1-2\alpha} dr)}  \leq  \frac{C_m}{1-\eta}  \int_0^\infty w(r) (w(r)^{-1} \V(r)-\tilde s)_+)\,  (1+ |\log r|)\, r \, dr, 
 $$
  where $\eta\in (0,1)$ has the same value as in \eqref{eta-1}, and $\tilde s= \eta\, c_h$. Proceeding  as above we thus get
 $$
  N( h_{m,1}^- - \V)_{L^2((0,1); \, r dr)} +N( h_{m,2}^- -  \V)_{L^2((1,\infty); \, r^{1-2\alpha} dr)}  \, \leq \,  \frac{C_m}{1-\eta}   \, a^{-1}\, \tilde s^{-a}\,  [ V]_{a}\, .
 $$
 Application of  inequality \eqref{2-aux} then completes the proof.
 \end{proof}

\begin{remark}
Note that Theorem \ref{pauli-long-range} is applicable, contrary to Propositions  \ref{prop-H-} and \ref{prop-H+}, also to slowly decaying potentials 
$V\in L^1(\R^2)$ such that $V\not\in L^1(\R^2,  |\log |x||\, dx)$. For example, for $V= W_\sigma$ we have $
[W_\sigma]_a < \infty $ if and only if $a\geq \sigma -1$, see \eqref{global}. Since $[\lambda V]_a = \lambda^{1+a}\,  [V]_a$ and since $ \| W_\sigma \|_{1,p}
=  \| W_\sigma \|_1$ for any $p>1$, setting $a= \sigma -1$ in 
\eqref{eq-long-range} gives 
$$
N( \PP-   \lambda W_\sigma) \, \leq\,  \alpha+1+  \lambda\, C_1(B)\,  \| V \|_1  + \lambda^\sigma \, C_2(B,a) \, [ W_\sigma]_{\sigma-1}\, . 
$$
By Proposition \ref{prop-semiclass-1}, this bound captures the correct behavior of $N( \PP -   \lambda W_\sigma)$ in the strong coupling limit.
\end{remark}


\section{\bf Magnetic Schr\"odinger operators}
\label{sec-schr}

\begin{proof}[\bf Proof of Corollary \ref{cor-main-schr}]
The positivity of the Pauli operator implies that, in the sense of quadratic forms on $H^1(\R^2)$, 
\begin{equation} \label{comparison}
2 (i\nabla +A)^2  \geq\, H_\pp. 
\end{equation} Hence
$$
N( (i\nabla +A)^2 -   V) \, \leq\,  N( H_\pp -  2 V) \, .
$$
Application of Proposition \ref{prop-H+} now completes the proof. 
\end{proof}


\begin{corollary}\label{cor-long-range}
Let $B$ satisfy \eqref{ass-B} and assume that $0 < \alpha \in\Z$. Then for any $p>1$ and any $a>0$ there exists a constants $C_1(B,p)$ and  $C_2(B,a)$  such that 
\begin{equation} 
N(  (i\nabla +A)^2-   V) \, \leq\,  C_1 (B,p)\,   \| V \|_{1,p} + C_2(B,a) \, [ V]_{a} 
\end{equation}
for all  $V \in L^1(\R_+, L^p(\Sph))$ for which the right hand side is finite.  
\end{corollary}

\begin{proof}
This follows from \eqref{comparison} and \eqref{long-range-h-plus}.
\end{proof}

\begin{remark}
If $\alpha=0$, then the arguments used in the proof of Corollaries \ref{cor-main-schr} and \ref{cor-long-range}  do not work because both operators $H_\pp$ and $H_\m$ are critical in this case. 
\end{remark}


\appendix

\section{\bf Positive definiteness of kernels}
\label{sec-app0}

\begin{lemma} \label{lem-positive-def_lem3-2}
Assume that $\V: (0,1) \to [0,\infty)$. Then the kernel  \begin{equation}
K(r,r') = - \sqrt{\V(r)}\ \sqrt{r r'}\,  \log(\max\{ r, r'\}) \,  \sqrt{\V(r')}
\end{equation}
is positive definite on $(0, 1)\times (0,1)$.

\end{lemma}

\begin{proof}
Let $N\in \N$, $r_1,\dots, r_N\in (0,1)$ and let $x_1,\dots, x_N\in\R$. Then, denoting
$$
y_j= x_j \sqrt{r_j \V(r_j)}, 
$$ 
we have
\begin{equation} 
\begin{aligned}
\sum_{j,k=1}^N x_j \, x_k \, K(r_j, r_k)  & = - \sum_{ j,k=1}^N \, y_j \, y_k \log(\max\{ r_j, r_k\}) = \sum_{ j,k=1}^N \, y_j \, y_k \log\Big(\min\big\{ \frac{1}{r_j}, \frac{1}{r_k}\big\} \Big) \\[4pt]
& =  \sum_{ j,k=1}^N \, y_j \, y_k \int_1^{\min\big\{ \frac{1}{r_j}, \frac{1}{r_k}\big\}}\ \frac 1t\, dt = \sum_{ j,k=1}^N \, y_j \, y_k \int_1^\infty \id_{(1, r_j^{-1})} (t) \ \id_{(1, r_k^{-1})} (t)\,  \frac 1t\, dt  \\[4pt]
& = \int_1^\infty  \sum_{ j,k=1}^N \, y_j \, y_k\  \id_{(1, r_j^{-1})} (t) \ \id_{(1, r_k^{-1})} (t)\,  \frac 1t\, dt   \\[4pt]
&=
\int_1^\infty \Big( \sum_{ j=1}^N   y_j \  \id_{(1, r_j^{-1})} (t) \Big)^2\,  \frac 1t\, dt  \geq 0.
\end{aligned}
\end{equation} 
\end{proof}

\begin{lemma} \label{lem-positive-def_lem3-3}
Let $\alpha>0$ and assume that $\V:\R_+ \to [0,\infty)$. Then the kernel  \begin{equation}
K(r,r') = \sqrt{ \V(r)}\  \frac{\sqrt{r r'}}{2\alpha} \left(\min \Big\{ \frac{r}{r'}, \frac{r'}{r} \Big\}\right)^\alpha  \sqrt{ \V(r')}
\end{equation}
is positive definite on $\R_+\times \R_+$.

\end{lemma}

\begin{proof}
Let $N\in \N$, $r_1,\dots, r_N\in (0,1)$ and let $x_1,\dots, x_N\in\R$. Denoting
$$
y_j= \frac{x_j}{(r_j)^\alpha} \sqrt{r_j \V(r_j)}, 
$$ 
we have
\begin{equation} \label{positive-def}
\begin{aligned}
\sum_{j,k=1}^N x_j \, x_k \, K(r_j, r_k)  & =  \frac{1}{2\alpha} \sum_{ j,k=1}^N \, y_j \, y_k \left(\min\Big\{ \frac{1}{r_j}, \frac{1}{r_k}\Big\}\right)^\alpha = \frac{1}{2\alpha}\sum_{ j,k=1}^N \, y_j \, y_k \int_0^{\min\big\{ \frac{1}{r_j}, \frac{1}{r_k}\big\}}\ \alpha t^{\alpha -1}\, dt   \\[4pt]
&  = \frac{1}{2} \int_0^\infty \Big( \sum_{ j=1}^N   y_j \  \id_{(0, r_j^{-1})} (t) \Big)^2\,  t^{\alpha -1} \, dt  \geq 0.
\end{aligned}
\end{equation} 
\end{proof}

\begin{lemma} \label{lem-pos-def-0}
Let $\alpha >0$ and assume that $\V: \R_+ \to [0,\infty)$. Then the kernel  $\sqrt{\V(r)}\ T_0^{-1}(r,r')\, \sqrt{\V(r')}$ is positive definite on $\R_+\times \R_+$.
\end{lemma}

\begin{proof} Let us split the kernel as follows; 
\begin{equation}
\sqrt{\V(r)}\ T_0^{-1}(r,r')\, \sqrt{\V(r')} = K_1(r,r') +K_2(r,r'), 
\end{equation}
where 
$$
K_2(r,r')= 
\begin{cases} 
- \sqrt{\V(r)}\ \sqrt{r r'}\,  \log(\max\{ r, r'\}) \,  \sqrt{\V(r')}  & \quad \text{if}\  0<r, r'\leq 1  \\[4pt] 
0 &\quad \text{otherwise } .
\end{cases}
$$
Now let  $r_1,\dots, r_N\in  \R_+$ and $x_1,\dots, x_N\in\R$ for some $N\in \N$.
By Lemma \ref{lem-positive-def_lem3-2},
\begin{equation} \label{K2-pos}
\sum_{j,k=1}^N x_j x_k \, K_2(r_j, r_k)  \, \geq \, 0.
\end{equation}
Hence it remains to prove the positivity of $K_1$. We define auxiliary functions $f,g: \R_+\to \R_+$ by
$$
f(s)=
\begin{cases}
			0 & \text{if}\  0 <s \leq  1 \,, \\[4pt]
			\alpha s^{\alpha-1} & \text{if}\    1 < s,
\end{cases}
\qquad \text{and} \qquad
g(s)=
\begin{cases}
			\alpha s^{\alpha-1} & \text{if}\  0 <s \leq  1 \,, \\[4pt]
			0 & \text{if}\    1 < s\, .
\end{cases}
$$
Then, denoting $y_j= x_j \sqrt{r_j \V(r_j)}$, we deduce from \eqref{limit-kernel-0} that
\begin{align*}
\sum_{j,k=1}^N x_j \, x_k \, K_1(r_j, r_k)  & = \frac{1}{2\alpha} \sum_{j,k=1}^N y_j \, y_k \, \Big(1+  \int_0^{\min\{ r_j, r_k\}} f(t)\, dt\Big)  \int_0^{\min\big\{ \frac{1}{r_j}, \frac{1}{r_k}\big\}} g(s)\, ds \\[2pt]
& =  \frac{1}{2\alpha} \int_0^\infty \!\! \sum_{j,k=1}^N y_j \, y_k \,  \id_{(0, r_j^{-1})} (s)\ \id_{(0, r_k^{-1})} (s)\, g(s)\,  ds \\[2pt]
& \quad +   \frac{1}{2\alpha} \int_0^\infty\!\!\!  \int_0^\infty \!\sum_{j,k=1}^N y_j \, y_k \, \id_{(0, r_j)} (t)\ \id_{(0, r_k)} (t)\,  \id_{(0, r_j^{-1})} (s)\ \id_{(0, r_k^{-1})} (s)\, f(t)\,  g(s)\,  ds dt\\[2pt]
& =   \frac{1}{2\alpha} \int_0^\infty\! \Big( \sum_{j=1}^N y_j  \id_{(0, r_j^{-1})} (s)\Big)^2 g(s)\, ds \\
& \quad +   \frac{1}{2\alpha} \int_0^\infty\!\!\!  \int_0^\infty\!\! \Big( \sum_{j=1}^N y_j\, \id_{(0, r_j)} (t) \id_{(0, r_j^{-1})} (s)\Big)^2 f(t)\, g(s)\, ds dt\\[2pt]
& \geq 0 .
\end{align*} 
In of view of \eqref{K2-pos} this completes the proof.
\end{proof}

\begin{lemma} \label{lem-pos-def-alphs}
Let $\alpha >0$ and assume that $\V: \R_+ \to [0,\infty)$. Then the kernel  $\sqrt{\V(r)}\ T_\alpha^{-1}(r,r')\, \sqrt{\V(r')}$ is positive definite on $\R_+\times \R_+$.
\end{lemma}

\begin{proof}
One can exploit that
\begin{align*}
T_\alpha^{-1}(r,r') = r r' \,T_0^{-1}\Big(\frac{1}{r},\frac{1}{r'}\Big).
\end{align*}
The claim then follows in the same way as in the proof of Lemma \ref{lem-pos-def-0}. Details are omitted. 
\end{proof}

\section{\bf One-dimensional weighted Hardy inequalities}
\label{sec-app1}

Here we recall some classical results on weighted Hardy inequalities. For their proofs we refer to \cite{muck}, see also \cite{flw}.

\begin{theorem}  \label{thm-class-1}
Let $U,W$ be nonnegative, a.e.-finite, measurable functions
on $\R_+$ such that
\begin{equation} 
\int_s^\infty U(t)^{-1}\, dt \, < \infty \qquad \forall\, s >0.
\end{equation}
Then for any locally absolutely continuous function $f$ on $\R_+$ with $\liminf_{t\to\infty} |f(t)|=0$ we have 
\begin{equation} 
\int_0^\infty W(t)\, |f(t)|^2\, dt \, \leq\, C(U,W) \, \int_0^\infty U(t)\, |f'(t)|^2\, dt,
\end{equation} 
where the constant $C(U,W)$ satisfies 
\begin{align} \label{C-upperb-1} 
C(U,W) & \leq 4 \, \sup_{s>0} \Big (\int_s^\infty U(t)^{-1}\, dt \Big)  \Big (\int_0^s W(t)\, dt \Big)  \, .
\end{align} 
\end{theorem} 

\medskip

\begin{corollary}  \label{cor-class-1}
Let $R\in \R_+$, and let $U,W$ be nonnegative, a.e.-finite, measurable functions
on $\R_+$ such that
\begin{equation} 
\int_s^\infty U(t)^{-1}\, dt \, < \infty \qquad \forall\, s >R.
\end{equation}
Then for any locally absolutely continuous function $f$ on $(0,R)$ with $\liminf_{t\to 	\infty} |f(t)|=0$ we have 
\begin{equation} 
\int_R^\infty W(t)\, |f(t)|^2\, dt \, \leq\, C(U,W) \, \int_R^\infty U(t)\, |f'(t)|^2\, dt,
\end{equation} 
where the constant $C(U,W)$ satisfies 
\begin{align} \label{C-upperb-5} 
C(U,W) & \leq 4 \, \sup_{R<s} \Big (\int_s^\infty U(t)^{-1}\, dt \Big)  \Big (\int_R^s W(t)\, dt \Big)  \, .
\end{align} 
\end{corollary}

\medskip

\begin{theorem}  \label{thm-class-2}
Let $U,W$ be nonnegative, a.e.-finite, measurable functions
on $\R_+$ such that
\begin{equation} 
\int_0^s U(t)^{-1}\, dt \, < \infty \qquad \forall\, s >0.
\end{equation}
Then for any locally absolutely continuous function $f$ on $\R_+$ with $\liminf_{t\to 0} |f(t)|=0$ we have 
\begin{equation} 
\int_0^\infty W(t)\, |f(t)|^2\, dt \, \leq\, C(U,W) \, \int_0^\infty U(t)\, |f'(t)|^2\, dt,
\end{equation} 
where the constant $C(U,W)$ satisfies 
\begin{align}  \label{C-upperb-3} 
C(U,W) & \leq 4 \, \sup_{s>0} \Big (\int_0^s U(t)^{-1}\, dt \Big)  \Big (\int_s^\infty W(t)\, dt \Big) \, .
\end{align} 
\end{theorem} 

\medskip

\begin{corollary}  \label{cor-class-2}
Let $R\in \R_+$, and let $U,W$ be nonnegative, a.e.-finite, measurable functions
on $\R_+$ such that
\begin{equation} 
\int_0^s U(t)^{-1}\, dt \, < \infty \qquad \forall\, s <R.
\end{equation}
Then for any locally absolutely continuous function $f$ on $(0,R)$ with $\liminf_{t\to 0} |f(t)|=0$ we have 
\begin{equation} 
\int_0^R W(t)\, |f(t)|^2\, dt \, \leq\, C(U,W) \, \int_0^R U(t)\, |f'(t)|^2\, dt,
\end{equation} 
where the constant $C(U,W)$ satisfies 
\begin{align} \label{C-upperb-2}
C(U,W) & \leq 4 \, \sup_{0<s<R} \Big (\int_0^s U(t)^{-1}\, dt \Big)  \Big (\int_s^R W(t)\, dt \Big)   \, .
\end{align} 
\end{corollary} 

\medskip

\section*{\bf Acknowledgements}
We are grateful to Rupert Frank, Dirk Hundertmark and Timo Weidl for numerous useful comments. 


\bibliographystyle{amsalpha}

\end{document}